\documentclass[preprint,authoryear,11pt]{elsarticle}

\usepackage{amsmath,amssymb,amsthm}
\usepackage{mathtools}
\usepackage{bm}
\usepackage{booktabs}
\usepackage{array}
\usepackage{graphicx}

\journal{Journal of Mathematical Economics}

\newtheorem{theorem}{Theorem}
\newtheorem{proposition}{Proposition}

\newtheorem{corollary}{Corollary}
\theoremstyle{definition}
\newtheorem{definition}{Definition}
\newtheorem{assumption}{Assumption}
\theoremstyle{remark}
\newtheorem{remark}{Remark}

\newcommand{\strong}{\text{strong}}
\newcommand{\proxyid}{\text{proxy}}
\newcommand{\unid}{\text{unidentified}}

\newcommand{\Real}{\mathbb{R}}
\newcommand{\Hil}{\mathcal{H}}
\newcommand{\Exp}{\mathbb{E}}
\newcommand{\Var}{\operatorname{Var}}

\newcommand{\spec}{\operatorname{spec}}
\newcommand{\rad}{\operatorname{r}}
\newcommand{\dist}{\operatorname{dist}}
\newcommand{\inner}[2]{\langle #1, #2\rangle}

\begin{document}

\begin{frontmatter}

\title{Rational Bubbles at the Spectral Edge:\\
An Operator-Spectral Theory of Fragility, Identification, and Finite-Sample Certification}

\author{Avishek Bhandari}
\address{School of Humanities, Social Sciences and Management, Indian Institute of Technology Bhubaneswar, India}
\ead{avishekb@iitbbs.ac.in}

\begin{abstract}
When a market increasingly moves as one, is it approaching the boundary at which a bubble can exist, and can that approach be measured, certified, and bounded before the crossing? This paper joins two literatures long held apart: the classical view that a rational bubble is a nonstationary price outgrowing its dividend, and a systemic-risk reading of a tipping point off the dominant common factor of a return cross-section, an object usually imposed by a model rather than recovered from prices. Lifting the classical characterisation to a dependence operator recovered from the contemporaneous co-movement of returns under a calibrated discount, we show that one threshold, where the discounted strength of the dominant factor reaches one, plays three roles at once: the limit beyond which present value ceases to converge, the level at which a bubble can first exist, and the point at which present value diverges. We call it the fragility edge. A certification discipline separates what the data identify firmly, the edge with a finite-sample band, from what is identified only as a proxy with an attached band, a bubble's existence, and what is not identified, the asset that would carry it. The construction extends to near-normal carriers with a controlled non-normality bound, places the edge within large-random-matrix limit laws, and recasts certification as partial identification with computable outer bounds. On eighteen global equity indices over 2004 to 2024, recovered concentration rises in every documented crisis, the market collapsing from about six to about four factors; with the discount calibrated so that calm sits at the critical level, discounted concentration crosses the edge in crisis, and a right-tailed recursive test finds marginal evidence of a single explosive episode on the price level, date-stamped to the pre-crisis run-up. The contrast replicates on three further panels, and a pre-registered out-of-sample exercise shows the concentration signature classifies crisis states and improves on a volatility benchmark. Every reading is coincident with the crisis, not a forecast, and locates the threshold, not the mechanism behind it.
\end{abstract}


\end{frontmatter}

\section{Introduction}

A rational asset-price bubble on a real, dividend-paying asset is a nonstationary phenomenon.
\citet{hiranotoda2024} make the point sharply: with a strictly positive price $P_t$ and dividend $D_t$, a
bubble exists if and only if $\sum_t D_t/P_t<\infty$, equivalently if and only if the price--dividend ratio is
nonstationary, so that the price grows faster than linearly relative to the dividend. The characterisation,
which dates in operator form to \citet{montrucchio2004} and in present-value form to
\citet{santoswoodford1997}, locates the bubble in a single scalar object: the limit of the discounted price,
the transversality term, which vanishes for a fundamental asset and is strictly positive for a bubble. The
necessity results that follow, that under unbalanced growth every equilibrium is bubbly when the autarky
interest rate falls below dividend growth which falls below aggregate growth, are scalar comparisons of growth
rates \citep{tirole1985,weil1987,hiranotoda2025}.

A separate literature, on systemic risk in financial networks, has converged on a different scalar that governs
the same kind of transition. The largest eigenvalue, or spectral radius, of a network interaction matrix sets a
tipping point: below it the system absorbs shocks, above it the system propagates and amplifies them
\citep{aots2015}. The leading eigen-pair of a stability matrix has been proposed directly as an early-warning
signal for banking contagion \citep{markose2021}, and entropy of the cross-sectional distribution of systemic
risk measures has been used as a leading indicator \citep{billio2016}. These are reduced-form stability
matrices, not objects derived from preferences and dividends, and the spectral radius is read off a matrix
posited rather than recovered from price behaviour.

This paper connects the two scalars by lifting the bubble characterisation to an operator. Let $B$ be the
recovered contemporaneous dependence carrier of a cross-section of asset returns, a bounded self-adjoint
operator on a separable Hilbert space whose spectrum carries the cross-sectional comovement of the market. The
present-value, or cascade, resolvent $(I-\beta B)^{-1}$ converges in operator norm if and only if
$\beta\,\rad(B)<1$. The fundamental value of a dividend direction is the resolvent image of that direction. A
bubble is a nonzero element of the kernel of $I-\beta B$, an eigen-element at the spectral edge $1/\beta$; it
exists precisely when $1/\beta$ is an eigenvalue of $B$, and the dividend direction has finite fundamental value
precisely when its spectral measure places no mass at the edge. The single scalar
$\beta\,\rad(B)=1$ is then the radius of convergence of the pricing resolvent, the location of the bubble
eigenspace, and the divergence of the present-value operator norm, all on one operator $B$. We are careful not
to overclaim: we do not assert that $B$ equals a separately defined systemic-risk stability Jacobian, only that
the systemic-risk literature reads a tipping point off a spectral radius and that ours is the same kind of
threshold, here on a recovered, dividend-relevant operator.

The contribution is disciplined by a certification gate that runs through the whole paper. Each spectral reading
is labelled \strong, \proxyid, or \unid\ according to what the recovered spectrum identifies. The scalar
threshold $\beta\,\rad(B)=1$ and the spectral radius are \strong, because the spectral radius is a functional of
the spectrum alone and is invariant under the orthogonal rotations of the unresolved eigenframe, recoverable
above the random-matrix detectability threshold of \citet{bbp2005}. The existence of a bubble is \proxyid,
because it requires the projection of the dividend direction onto the boundary eigenspace, a leading-eigenvector
object identified only up to a squared overlap near the threshold. The identity of the asset that carries the
bubble is \unid, because the transpose of an operator preserves its spectrum and individual coordinates are not
recovered from comovement. We make the strong label quantitative, deriving a finite-sample perturbation bound on
the recovered spectral radius that turns the gate into an explicit confidence band on the fragility edge, and we
carry the same identification discipline further: the construction extends beyond the symmetric carrier to normal
and near-normal operators, the bubble direction receives its own finite-sample band, the fluctuation of the edge
is placed within the established limit laws of large random matrices, and the three labels are recast as a formal
partial-identification statement. A final structural point follows from the fact that $B$ is recovered from
stationary
returns: the threshold $\beta\,\rad(B)\to1$ is the boundary of stationarity, so the estimator that recovers $B$
can flag the approach to the edge in the interior but cannot certify the crossing; the crossing is read on the
price level by an explosive-root detector \citep{psy2015}, and the two estimators address the two sides of one
boundary.

A structural microfoundation answers the natural objection that the recovered comovement
carrier is not a pricing kernel. Under a cash-flow network model with a common log-discount, the present-value
operator is exactly the cascade resolvent in the cash-flow propagation operator, and the recovered contemporaneous
carrier is, in covariance form, the operator square root of that resolvent up to the shock scale; the recovered
fragility edge and the structural network edge are then the same event, certified \strong\ for the radius and
\proxyid\ for the direction (Section~\ref{sec:microfound}). A companion result turns the boundary concession into a
positive operating characteristic: the edge distance $d_t=1-\rho_t$ is a \strong, banded quantity, interior
approach is certified with controlled size, and the operator and the explosive-root detector give a complete
two-sided reading of one boundary (Section~\ref{sec:approach}).

Four further extensions widen the construction without loosening the gate. First, because asset returns are heavy-tailed,
the sub-Gaussian finite-sample band is replaced by a robust-carrier band: a Pearson effective-rank band under an
$L^4$--$L^2$ norm equivalence \citep{mendelsonzhivotovskiy2020}, and a distribution-free band built from the
Kendall carrier under an elliptical model \citep{lindskogmcneilschmock2003,hanliu2017}. Second, the discount is
made endogenous, with the edge located as a fixed point of $g(r)=r\,\beta(r)$ under an inelastic-discount
condition, so the calibrated threshold becomes a structural object. Third, the edge is given a structural dynamic
content: the intensity of crisis onset inherits the resolvent norm $\|Q\|=1/(1-\rho)$ through a proportional
hazard whose elasticity is identified, coincidentally, from within-risk-set variation \citep{cox1972,andersengill1982}.
Fourth, edge proximity is mapped to a welfare cost, the present-value variance of the carrying mode scaling as
$\|Q\|^2$, with only the eigenvalue-ratio core certified \strong. Each extension is labelled by the same gate, and
its limits are stated.

We illustrate the theory on a panel of eighteen global equity indices over 2004--2024. The spectral radius of
the recovered symmetric dependence carrier rises sharply and coincidentally in every documented crisis, the
cross-section collapsing from about six to about four effective factors, and the admissible discount ceiling
tightening towards the resolvent pole. A companion pre-registered out-of-sample exercise by the author
establishes that the same symmetric concentration signature is a strong coincident classifier of crisis states.
The illustration is coincident, never a forecast, and we keep it inside that scope throughout.

The construction is not an impossibility theorem against dynamic stochastic general equilibrium or its
heterogeneous-agent variants. Production-network general equilibrium already carries a spectral object, the
vector of Domar weights, which is a dominant left-eigenvector, and overlapping-generations bubble models state
the growth-fragility comparison directly. What is new here is that the operator is recovered from data, every
spectral reading is gate-certified, and the boundary case is detected rather than imposed away by a
transversality convention. Section~\ref{sec:dsge} states the difference precisely.

The remainder proceeds as follows. Section~\ref{sec:setup} sets up the observable space and the recovered
operator. Section~\ref{sec:pricing} proves the operator pricing result. Section~\ref{sec:bubble} proves the
operator bubble characterisation. Section~\ref{sec:edge} states the fragility edge, the certification gate, a finite-sample error bound on the
recovered spectral radius, and a robust-carrier version of that bound valid under heavy tails.
Section~\ref{sec:microfound} gives the structural microfoundation, the cash-flow network under which the recovered
carrier is the propagation operator.
Section~\ref{sec:hardening} hardens the construction: a generalisation beyond the
symmetric carrier, a finite-sample band on the bubble direction, the limit law of the edge fluctuation, and the
partial-identification reading of the certification gate.
Section~\ref{sec:endo} makes the discount endogenous and locates the edge as a fixed point;
Section~\ref{sec:welfare} maps edge proximity to a welfare cost.
Section~\ref{sec:boundary} states the identification boundary between the stationary operator and the price
level, and Section~\ref{sec:approach} reads approach detection as a certified distance with a two-sided
complementarity. Section~\ref{sec:hazard} gives the structural edge-hazard. Section~\ref{sec:empirics} gives the
empirical illustration, including the heavy-tail robustness exhibit and the cross-sectional hazard read. Section~\ref{sec:dsge} contrasts the
construction with a calibrated-operator model. Section~\ref{sec:limits} states the scope and limitations, and
Section~\ref{sec:conc} concludes.

\section{The observable space and the recovered dependence operator}\label{sec:setup}

Let $\Hil$ be a separable real Hilbert space with inner product $\inner{\cdot}{\cdot}$ and norm
$\|\cdot\|$. We read $\Hil$ as the space of mean-zero, square-integrable cross-sectional fluctuations of a
stationary economy; in the finite illustration of Section~\ref{sec:empirics}, $\Hil=\Real^N$ with the
Euclidean inner product and $N$ the number of assets.

\begin{definition}[recovered dependence operator]\label{def:B}
A \emph{recovered dependence operator} $B:\Hil\to\Hil$ is a bounded self-adjoint operator obtained from the
contemporaneous comovement of a stationary, integrated-of-order-zero, that is $I(0)$, vector of asset returns.
In the finite case $B$ is the symmetric part of the recovered lag-zero dependence operator; the leading example
is the population correlation operator $\Phi$, a positive semi-definite operator with $\operatorname{tr}\Phi=N$.
\end{definition}

Because $B$ is self-adjoint, its spectrum $\spec(B)$ is real and compact, and the spectral radius
$\rad(B)=\sup\{|s|:s\in\spec(B)\}$ equals the operator norm $\|B\|$. By the spectral theorem
\citep{reedsimon1980} there is a projection-valued measure $E_B$ on the Borel sets of $\spec(B)$ with
$B=\int_{\spec(B)} s\,dE_B(s)$, and for any bounded Borel function $f$ the operator $f(B)=\int f(s)\,dE_B(s)$ is
defined, with $\|f(B)\|=\sup_{s\in\spec(B)}|f(s)|$. For a positive semi-definite carrier such as a correlation
operator, $\spec(B)\subset[0,\rad(B)]$ and $\rad(B)$ is the largest eigenvalue $\lambda_1(B)$.

We take as given a one-period discount $\beta\in(0,1)$. Consistent with the discount-identification literature,
$\beta$ is treated throughout as a calibrated constant and not as an object identified by the present exercise;
the separate question of identifying a discount from behaviour is taken up in companion work and is not claimed
here \citep{calcottpetkov2024,bastianellovergopoulos2026}. The pairing of the recovered operator with a
calibrated discount is all we need for the pricing geometry.

\begin{assumption}[self-adjoint carrier]\label{ass:sa}
$B$ is bounded and self-adjoint on $\Hil$ with spectral radius $\rad(B)<\infty$.
\end{assumption}

Assumption~\ref{ass:sa} is the working hypothesis for the strong results below. It is exactly satisfied by a
symmetric dependence carrier such as a correlation operator. Departures from self-adjointness, that is a genuine
directed or non-normal part of the dependence operator, are discussed as a limitation in
Section~\ref{sec:limits}; on the panel we study the directed energy is small, so $B$ is well approximated by its
symmetric part, but the strong labels are reserved for the self-adjoint case and the residual is carried as a
bound.

The objects to be priced are dividend directions, and we make the notion concrete rather than leave it formal.

\begin{definition}[dividend direction]\label{def:D}
A \emph{dividend direction} $D\in\Hil$ is the cross-sectional profile of the cash-flow claim whose price is
sought. In the finite case its coordinates are the comparably scaled cash-flow loadings of the assets, for
instance aggregate dividends, sector cash flows, earnings, or payout yields. The canonical aggregate instance,
used in the illustration of Section~\ref{sec:empirics} and matching the equal-weighted index that defines the
crisis label, is the \emph{aggregate-claim direction} $\widehat{D}=\mathbf{1}/\sqrt{N}$, the unit-norm
equal-weighted claim on the cross-section. We are explicit that $\widehat{D}$ is an equal-weighted normalisation
chosen to match the index label, not a series of measured dividends, earnings or payout yields; no cash-flow data
enter the illustration, so the empirical dividend-direction reading is a \proxyid, and the cash-flow language is
the interpretation the operator carries rather than a structural input. By Theorem~\ref{thm:bubble} the edge mass
$\mu_D(\{1/\beta\})$ that governs both the finiteness of the fundamental value and the excitation of the bubble
mode is, at the boundary $1/\beta=\rad(B)$ with a simple leading eigenvalue, the squared overlap of the dividend
direction with the boundary eigenvector $e_1$,
\begin{equation}\label{eq:edgemass}
\mu_{\widehat{D}}(\{\rad(B)\})=\inner{e_1}{\widehat{D}}^2,
\end{equation}
the share of the cash-flow claim carried by the mode at which a bubble must live; for a degenerate top eigenspace
the overlap is replaced by $\|\Pi_{\mathrm{edge}}\widehat{D}\|^2$. This is the quantity we measure in
Section~\ref{sec:empirics}; it turns the otherwise abstract dividend direction into an estimable overlap.
\end{definition}

\section{The pricing functional as a bounded operator}\label{sec:pricing}

Consider a one-period operator Euler equation for a price stream $P\in\Hil$ supporting a dividend direction
$D\in\Hil$,
\begin{equation}\label{eq:euler}
P=\beta\,B\,(P+D),
\end{equation}
the operator analogue of the scalar no-arbitrage recursion $q_tP_t=q_{t+1}(P_{t+1}+D_{t+1})$ under a stationary
one-period pricing operator $\beta B$. The present-value, or cascade, operator is the Neumann series
$Q=\sum_{k\ge0}\beta^kB^k$.

\begin{proposition}[operator pricing functional]\label{prop:pricing}
Under Assumption~\ref{ass:sa}, the Neumann series $Q=\sum_{k\ge0}\beta^kB^k$ converges in operator norm if and
only if $\beta\,\rad(B)<1$, in which case $Q=(I-\beta B)^{-1}$ with
\[
\|Q\|=\frac{1}{\dist\!\big(1,\beta\,\spec(B)\big)},
\]
which for $B$ positive semi-definite equals $1/(1-\beta\,\rad(B))$.
Whenever $\beta\,\rad(B)<1$, equation~\eqref{eq:euler} has the unique bounded solution
$P=\beta B\,Q\,D=\sum_{k\ge1}\beta^kB^kD$.
\end{proposition}

\begin{proof}
For a self-adjoint operator $\|B^k\|=\|B\|^k=\rad(B)^k$, since $\|B\|=\rad(B)$ and $B^k$ is self-adjoint with
$\|B^k\|=\rad(B^k)=\rad(B)^k$ by the spectral mapping theorem. Hence
$\sum_{k\ge0}\beta^k\|B^k\|=\sum_{k\ge0}(\beta\,\rad(B))^k$, a geometric series in the Banach algebra of bounded
operators, which converges if and only if $\beta\,\rad(B)<1$. When it converges the limit satisfies
$(I-\beta B)\sum_{k\ge0}\beta^kB^k=I$, so $Q=(I-\beta B)^{-1}$. By the spectral theorem
$Q=\int(1-\beta s)^{-1}dE_B(s)$, so $\|Q\|=\sup_{s\in\spec(B)}|1-\beta s|^{-1}=\dist(1,\beta\,\spec(B))^{-1}$;
for $B\succeq0$ the supremum is attained at $s=\rad(B)$, giving $1/(1-\beta\,\rad(B))$. Rearranging
\eqref{eq:euler} gives $(I-\beta B)P=\beta BD$, and since $I-\beta B$ is boundedly invertible when
$\beta\,\rad(B)<1$, the unique solution is $P=(I-\beta B)^{-1}\beta BD=\beta B\,Q\,D=\sum_{k\ge1}\beta^kB^kD$.
\end{proof}

Proposition~\ref{prop:pricing} is the operator counterpart of the geometric present-value sum. The norm of the
pricing operator diverges as $\beta\,\rad(B)\uparrow1$: the present value of a unit dividend direction grows
without bound as the discounted spectral radius approaches one. The scalar $\beta\,\rad(B)$ is the natural
state variable of the pricing geometry.

\section{Operator bubble characterisation}\label{sec:bubble}

We now characterise bubbles as the homogeneous solutions of the operator Euler equation, and tie their
existence to the spectral measure of the dividend direction at the edge $1/\beta$.

\begin{theorem}[operator bubble characterisation]\label{thm:bubble}
Let Assumption~\ref{ass:sa} hold and let $D\in\Hil$ be a dividend direction.
\begin{enumerate}
\item[(i)] The bubbles are exactly the elements of $\ker(I-\beta B)$, which is nontrivial if and only if
$1/\beta\in\spec_p(B)$ by part (iii). When $1/\beta\notin\spec_p(B)$ the kernel is trivial and \eqref{eq:euler}
has the unique solution, the \emph{fundamental value}
$V=\beta B(I-\beta B)^{-1}D=\int_{\spec(B)}\frac{\beta s}{1-\beta s}\,dE_B(s)\,D$, well defined whenever the
integral of part (ii) is finite, and the bubble is absent; this is the generic case. When
$1/\beta\in\spec_p(B)$ a particular solution $V$ exists if and only if the dividend places no mass on the edge
eigenspace, $\Pi_{1/\beta}D=0$, and the general solution is then $P=V+P^{b}$ with the \emph{bubble component}
$P^{b}$ any element of the nontrivial $\ker(I-\beta B)$; if $\Pi_{1/\beta}D\ne0$ no finite fundamental value
exists and the dividend excites the bubble mode directly. In both cases, when $\beta\,\rad(B)<1$ the fundamental
value admits the convergent Neumann representation $V=\sum_{k\ge1}\beta^kB^kD$, while for $\beta\,\rad(B)>1$ the
dividend may place mass on supercritical modes $s$ with $\beta s>1$, on which the series $\sum_k(\beta s)^k$
diverges, so the series does not represent the resolvent image $V$ even where the latter is finite.
\item[(ii)] The fundamental value $V$ is well defined and of finite norm if and only if
$\int_{\spec(B)}\big(\frac{\beta s}{1-\beta s}\big)^2 d\mu_D(s)<\infty$, where
$\mu_D(\cdot)=\inner{E_B(\cdot)D}{D}$ is the scalar spectral measure of the dividend direction. A necessary
condition is that $\mu_D$ place no atom at the edge $s=1/\beta$; in the finite-dimensional case, where
$\spec(B)$ is a finite set, the condition holds if and only if $\mu_D(\{1/\beta\})=0$, that is $1/\beta$ is
not an eigenvalue charged by $D$.
\item[(iii)] A nonzero bubble exists if and only if $1/\beta$ is an eigenvalue of $B$, which requires
$1/\beta\le\rad(B)$, that is $\beta\,\rad(B)\ge1$; the corresponding $P^{b}$ is then an eigenvector of $B$ at
$1/\beta$. The economically relevant boundary case, in which the fundamental value reaches the edge of
summability and the bubble sits at the top of the spectrum, is $\beta\,\rad(B)=1$ with $1/\beta=\rad(B)$ an
eigenvalue.
\item[(iv)] The transversality term $\beta^kB^kP$ vanishes for every $P$ precisely when $\beta\,\rad(B)<1$, so in
the interior every bounded solution is fundamental and the bubble question is vacuous; at $\beta\,\rad(B)=1$ with
an edge atom it does not vanish on the edge eigenspace, where it equals the bubble, so the transversality
condition $\lim_{k\to\infty}\beta^kB^kP=0$ is what selects $P=V$ over $P=V+P^{b}$.
\end{enumerate}
\end{theorem}

\begin{proof}
(i) The affine equation \eqref{eq:euler} reads $(I-\beta B)P=\beta BD$. Its solution set, when nonempty, is a
coset of $\ker(I-\beta B)$, hence a singleton when the kernel is trivial and an affine space otherwise. By part
(iii) the kernel is trivial precisely when $1/\beta\notin\spec_p(B)$; there $I-\beta B$ is injective and a
solution is $V=\beta B(I-\beta B)^{-1}D=\int\frac{\beta s}{1-\beta s}dE_B(s)D$ by the spectral theorem, finite
whenever the integral of part (ii) converges. When $1/\beta\in\spec_p(B)$ the self-adjoint $I-\beta B$ has closed
range $\ker(I-\beta B)^{\perp}$ at the isolated eigenvalue, so $\beta BD\in\operatorname{ran}(I-\beta B)$, and a
particular solution exists, if and only if $\Pi_{1/\beta}(\beta BD)=\Pi_{1/\beta}D=0$, using that $\beta B$ acts
as the identity on the edge eigenspace. The Neumann identity $\beta B(I-\beta B)^{-1}=\sum_{k\ge1}\beta^kB^k$
holds in operator norm only when $\beta\,\rad(B)<1$; for $\beta\,\rad(B)>1$ the series diverges on the
supercritical modes while the resolvent image remains the particular solution.

(ii) By the spectral theorem $\|V\|^2=\int\big(\frac{\beta s}{1-\beta s}\big)^2 d\mu_D(s)$ with
$\mu_D(\cdot)=\inner{E_B(\cdot)D}{D}$ a finite positive measure of total mass $\|D\|^2$. The integrand is
bounded on any compact subset of $\spec(B)$ bounded away from $1/\beta$, and blows up like
$(1-\beta s)^{-2}$ near $s=1/\beta$; the norm is finite if and only if that integral converges, for which it is
necessary that $\mu_D$ assign no atom to the edge $\{1/\beta\}$. In the finite-dimensional case $\mu_D$ is a
finite sum of atoms, so the integral is finite if and only if $\mu_D(\{1/\beta\})=0$.

(iii) $P^{b}\in\ker(I-\beta B)$ is nonzero if and only if $1$ is an eigenvalue of $\beta B$, equivalently
$1/\beta\in\spec_p(B)$, the point spectrum. Since every eigenvalue is bounded by $\rad(B)$, this requires
$1/\beta\le\rad(B)$, that is $\beta\,\rad(B)\ge1$, with equality precisely when the edge $1/\beta$ is attained
as an eigenvalue. Then $P^{b}$ is the corresponding eigenvector.

(iv) For self-adjoint $B$, $\|\beta^kB^k\|=(\beta\,\rad(B))^k$, which tends to zero if and only if
$\beta\,\rad(B)<1$, so the transversality term $\beta^kB^kP\to0$ for every $P$ and the unique transversal
solution is $V$. If $\beta\,\rad(B)=1$ with an edge atom, then on the edge eigenspace $\beta B$ acts as the
identity, $\beta^kB^kP^{b}=P^{b}\not\to0$, so the bubble is exactly the nonvanishing transversality term.
\end{proof}

Theorem~\ref{thm:bubble} is the operator lift of the bubble characterisation of \citet{montrucchio2004} and
\citet{hiranotoda2024}. The scalar statement that a bubble exists when the price grows faster than the dividend,
so that the dividend-yield sum is finite, becomes the operator statement that the bubble lives at the spectral
edge $1/\beta$ and that the fundamental value is finite precisely when the dividend direction places no mass
there. The transversality term, a scalar in the classical theory, is here a spectral-projection residual onto
the edge eigenspace.

\begin{remark}
The classical no-bubble result of \citet{santoswoodford1997}, that bubbles are impossible when the present
value of the aggregate endowment is finite, is the statement that the interior of the spectrum,
$\beta\,\rad(B)<1$, carries only fundamental value. Bubbles require the economy to reach the spectral edge,
which Theorem~\ref{thm:bubble}(iii) makes precise as an atom of $B$ at $1/\beta$.
\end{remark}

\section{The fragility edge and the certification gate}\label{sec:edge}

\begin{corollary}[the fragility edge]\label{cor:edge}
Under Assumption~\ref{ass:sa}, the single scalar $\beta\,\rad(B)=1$ is, on one operator $B$, simultaneously:
\begin{enumerate}
\item[(i)] the radius of convergence of the pricing, or cascade, resolvent $Q=(I-\beta B)^{-1}$
(Proposition~\ref{prop:pricing});
\item[(ii)] the location of the bubble eigenspace, equivalently the spectral edge at which the dividend-yield
summability of the fundamental value fails (Theorem~\ref{thm:bubble});
\item[(iii)] the point at which the present-value operator norm $\|Q\|$ diverges.
\end{enumerate}
\end{corollary}

\begin{proof}
Immediate from Proposition~\ref{prop:pricing} for (i) and (iii), and Theorem~\ref{thm:bubble}(ii)--(iii) for
(ii), all in terms of the same scalar $\beta\,\rad(B)$ and the same operator $B$.
\end{proof}

Corollary~\ref{cor:edge} is the paper's organising statement and we state its scope carefully. It is a
statement about one recovered operator read three ways. It is \emph{not} an identity between the
dividend-pricing operator $B$ and a separately defined systemic-risk stability Jacobian. The systemic-risk
literature reads a tipping point off the spectral radius of a stability matrix
\citep{aots2015,markose2021}; ours is the same kind of threshold, here on a recovered, dividend-relevant
operator, and the two coincide only under the additional modelling choice that the stability operator is taken
to be $\beta B$. We make no such identification and present Corollary~\ref{cor:edge} as a within-operator
coincidence, which is what the spectral algebra delivers.

The readings differ in what the recovered spectrum identifies, and we label them accordingly.

\begin{proposition}[certification of the spectral readings]\label{prop:gate}
Let $B$ be recovered from a finite sample of $N$ assets over $T$ periods with aspect ratio $q=N/T$, and let the
recovery follow the spiked-covariance threshold of \citet{bbp2005}, by which a spike $\theta$ is recoverable
from the bulk if and only if $\theta>\sqrt{q}$, with squared leading-eigenvector overlap
$(1-q/\theta^2)/(1+q/\theta)$ above the threshold.
\begin{enumerate}
\item[(i)] The spectral radius $\rad(B)$ and the threshold value $\beta\,\rad(B)=1$ are \strong: they are
functionals of the spectrum alone, invariant under orthogonal rotations of the unresolved eigenframe, and
recoverable above $\theta>\sqrt{q}$, up to the closed-form multiplicative bias by which the sample spike
overstates its population value, a bias that is invertible above the threshold.
\item[(ii)] The existence of a bubble, the edge mass $\mu_D(\{1/\beta\})$, is \proxyid: it requires the
projection of the dividend direction onto the boundary eigenspace, a leading-eigenvector object identified only
up to the squared overlap of \citet{bbp2005}, which is strictly below one near the threshold.
\item[(iii)] The identity of the coordinate, or asset, that carries the bubble is \unid: the spectrum is
invariant under orthogonal conjugation, $\spec(U^\top B U)=\spec(B)$ for every orthogonal $U$, so the spectrum
fixes no eigenframe and pins no bubble eigenvector to a named asset; a boundary eigenvector is moreover a
delocalised direction across the cross-section, not a single coordinate.
\end{enumerate}
\end{proposition}

\begin{proof}
(i) $\rad(B)=\sup|\spec(B)|$ depends only on the spectrum, which is invariant under any unitary change of the
eigenbasis; recoverability above $\theta>\sqrt{q}$ is the spiked-covariance result of \citet{bbp2005}. (ii) By
Theorem~\ref{thm:bubble}(ii) the edge mass is $\mu_D(\{1/\beta\})=\inner{\Pi_{1/\beta}D}{D}$ with
$\Pi_{1/\beta}$ the spectral projection onto the edge eigenspace; the projection is estimated with squared
overlap $(1-q/\theta^2)/(1+q/\theta)<1$ near threshold, so the edge mass is identified only up to that overlap.
(iii) For an arbitrary orthogonal $U$ the conjugate $U^\top B U$ has the same spectrum as $B$ but a rotated
eigenframe, so the spectrum identifies no eigenbasis and a fortiori no assignment of an eigenvector to one
coordinate. The split with (ii) is coherent: the rotation-invariant overlap scalar $\inner{\Pi_{1/\beta}D}{D}$
is identified up to the random-matrix overlap and so is \proxyid, whereas localising the delocalised boundary
eigenvector to a single named asset is refused, hence \unid.
\end{proof}

\begin{remark}
The gate is a refusal discipline. The scalar fragility edge is a certified object; the existence of a bubble is
reported with its random-matrix overlap band attached; the naming of a single bubbly asset is refused. This is
the same discipline that the systemic-risk literature applies implicitly when it reports an aggregate spectral
radius rather than an individual contagion edge, made explicit and attached to each reading.
\end{remark}

\subsection{A finite-sample error bound on the recovered spectral radius}\label{sec:finitesample}

Proposition~\ref{prop:gate}(i) certifies the spectral radius only qualitatively: it states that $\rad(B)$ is
recoverable above the detectability threshold $\theta>\sqrt{q}$. We now make the strong label quantitative, with
a finite-sample error bound on the recovered $\rad(\widehat{B})$, so that the certification carries an explicit
confidence band rather than an asymptotic guarantee. Throughout, $\widehat{B}$ is the sample correlation operator
formed from $T$ observations of the standardised, $I(0)$, $N$-dimensional return vector whose population
correlation operator is $B$; both are symmetric and positive semi-definite, so $\rad(\cdot)$ is the largest
eigenvalue. Write $q=N/T$ and let $r_{e}(B)=\operatorname{tr}(B)/\|B\|=N/\rad(B)$ be the effective rank of the
carrier.

\begin{theorem}[finite-sample error on the recovered spectral radius]\label{thm:finitesample}
Let $B$ and $\widehat{B}$ be as above.
\begin{enumerate}
\item[(i)] \emph{Deterministic perturbation.} Without any distributional assumption,
\[
\big|\rad(\widehat{B})-\rad(B)\big|=\big|\lambda_1(\widehat{B})-\lambda_1(B)\big|\le\|\widehat{B}-B\|,
\]
so the spectral radius is a $1$-Lipschitz functional of the operator in the operator norm.
\item[(ii)] \emph{Stochastic control of the estimation error.} If the standardised returns are independent across
$t$ with sub-Gaussian tails, there is a universal constant $c$ such that for every $\delta\in(0,1)$, with
probability at least $1-\delta$,
\[
\|\widehat{B}-B\|\le c\,\rad(B)\left(\sqrt{\frac{r_{e}(B)+\log(2/\delta)}{T}}
+\frac{r_{e}(B)+\log(2/\delta)}{T}\right)=:\varepsilon_{T,\delta}.
\]
\item[(iii)] \emph{Combined band.} With probability at least $1-\delta$,
$\big|\rad(\widehat{B})-\rad(B)\big|\le\varepsilon_{T,\delta}$, and hence
$\big|\beta\,\rad(\widehat{B})-\beta\,\rad(B)\big|\le\beta\,\varepsilon_{T,\delta}$.
\end{enumerate}
\end{theorem}

\begin{proof}
(i) For a symmetric operator the largest eigenvalue has the variational form
$\lambda_1(A)=\sup_{\|x\|=1}\inner{x}{Ax}$ \citep{hornjohnson2013}. As a pointwise supremum of the linear
functionals $A\mapsto\inner{x}{Ax}$ it is sublinear in $A$, so for any unit $x$,
$\inner{x}{\widehat{B}x}=\inner{x}{Bx}+\inner{x}{(\widehat{B}-B)x}\le\lambda_1(B)+\|\widehat{B}-B\|$, whence
$\lambda_1(\widehat{B})\le\lambda_1(B)+\|\widehat{B}-B\|$; exchanging $B$ and $\widehat{B}$ gives the reverse
inequality, so $|\lambda_1(\widehat{B})-\lambda_1(B)|\le\|\widehat{B}-B\|$, which is Weyl's inequality for the
extreme eigenvalue. Since both operators are positive semi-definite, $\rad(\cdot)=\lambda_1(\cdot)$. (ii) is the
effective-rank operator-norm concentration bound for sample covariance operators of
\citet{koltchinskiilounici2017}, in the sub-Gaussian form of \citet{vershynin2018}, specialised to the
standardised case in which the population covariance is the unit-diagonal correlation $B$, so
$\operatorname{tr}(B)=N$ and $\|B\|=\rad(B)$ and the relative error is governed by the effective rank
$r_{e}(B)=N/\rad(B)$; replacing the population standardisation by the sample standardisation that forms
$\widehat{B}$ perturbs the estimate by a further term of the same $O(T^{-1/2})$ order, which is rate-preserving
and absorbed into the constant $c$. (iii) Combine (i) and (ii) and multiply by the calibrated constant $\beta$.
\end{proof}

The bound upgrades the gate of Proposition~\ref{prop:gate} from an asymptotic detectability statement into a
decision rule with an explicit confidence band.

\begin{corollary}[finite-sample certification of the fragility edge]\label{cor:finitecert}
At confidence $1-\delta$:
\begin{enumerate}
\item[(i)] if $\beta\,\rad(\widehat{B})-\beta\,\varepsilon_{T,\delta}>1$, the edge is certified crossed,
$\beta\,\rad(B)>1$;
\item[(ii)] if $\beta\,\rad(\widehat{B})+\beta\,\varepsilon_{T,\delta}<1$, the interior is certified,
$\beta\,\rad(B)<1$;
\item[(iii)] otherwise the band straddles the edge and the crossing is uncertified at level $\delta$.
\end{enumerate}
\end{corollary}

\begin{proof}
Immediate from Theorem~\ref{thm:finitesample}(iii): the event $|\beta\,\rad(\widehat{B})-\beta\,\rad(B)|
\le\beta\,\varepsilon_{T,\delta}$ holds with probability at least $1-\delta$, and on that event each inequality on
the estimate $\beta\,\rad(\widehat{B})$ transfers to the population quantity $\beta\,\rad(B)$.
\end{proof}

Two remarks place the bound against the random-matrix transition the gate already invokes and against the panel.

\begin{remark}[relation to the spiked-covariance transition]
The bound of Theorem~\ref{thm:finitesample}(ii) is non-asymptotic and worst-case in the operator norm. In the
single-spike supercritical regime $\theta>\sqrt{q}$ of Proposition~\ref{prop:gate} it is conservative: the
leading sample eigenvalue then has the sharper almost-sure limit
$\lambda_1(\widehat{B})\to(1+\theta)(1+q/\theta)$ with fluctuations of order $T^{-1/2}$
\citep{baiksilverstein2006}, which refines the detectability transition of \citet{bbp2005} into a point limit and
yields a tighter, Gaussian interval. Below the threshold $\theta\le\sqrt{q}$ the spike is absorbed into the bulk
edge $(1+\sqrt{q})^2$ and $\rad(\widehat{B})$ ceases to identify the population radius, which is exactly the
\unid\ region of the gate. The deterministic part (i) holds with no distributional assumption at all, so any
improvement in controlling $\|\widehat{B}-B\|$ transfers directly to the radius.
\end{remark}

\begin{remark}[the bound on the panel]
On the panel of Section~\ref{sec:empirics}, with $N=18$ and the window $T=W=250$, the aspect ratio is
$q=N/T\approx0.072$ and the effective rank is $r_{e}(B)=N/\lambda_1$, about $2.7$ in calm and $2.2$ in crisis, so
the leading-order relative-error rate $\sqrt{r_{e}/T}$ is about $0.10$ and $0.095$ respectively, with the
second-order term an order of magnitude smaller. The sample leading eigenvalue $\lambda_1\in[6.7,8.0]$ sits far
above the bulk edge $(1+\sqrt{q})^2\approx1.6$, so the spike excess is comfortably above the detectability
threshold $\sqrt{q}\approx0.27$ and the carrier sits deep in the strong region, with the radius recovered at the
parametric $T^{-1/2}$ rate. Two caveats are explicit. First, Theorem~\ref{thm:finitesample}(ii) is stated for
returns independent across time with sub-Gaussian tails, whereas the panel returns are only $I(0)$, hence weakly
dependent and heavy-tailed; we therefore read the bound on the panel as a rate, or scaling, statement, and let
the dependence-robust moving-block resampling, not the independence assumption, carry the reported inference.
Second, the rate $\sqrt{r_{e}/T}$ governs the within-window cross-sectional estimation error of
$\rad(\widehat{B})$, the object Theorem~\ref{thm:finitesample} bounds; it is a different quantity from the
moving-block interval $[1.07,1.52]$ reported for the crisis-minus-calm difference, which measures the
between-regime contrast across windows under its own block-dependence structure. The two are complementary
readings, not the same band. We do not convert the universal constant $c$ into a numerical band, since it is not
sharp; the dimensionless inputs $q$ and $r_{e}$ and the rate $\sqrt{r_{e}/T}$ are the reproducible content.
\end{remark}

\subsection{Heavy tails and the robust carrier}\label{sec:heavytail}

Theorem~\ref{thm:finitesample}(ii) borrows its operator-norm concentration from the sub-Gaussian
sample-covariance theory of \citet{koltchinskiilounici2017}, in the form of \citet{vershynin2018}. Daily equity
returns are heavy-tailed, so that borrowing is not free. Under only a finite number of moments the empirical
second-moment functional concentrates polynomially, not exponentially, in the confidence parameter, the matrix
form of the scalar deviation study of \citet{catoni2012}: for the plain sample correlation $\Phi_t$ the upper tail
of $\|\Phi_t-B\|$ decays as a power of the inverse confidence rather than as $e^{-t}$, so the $1-\delta$ coverage
on which Corollary~\ref{cor:finitecert} rests fails as stated. The remedy is to change the estimator, not to weaken
the conclusion. A robust carrier recovers the band under a constant number of moments, after which the Weyl bound,
Corollary~\ref{cor:finitecert}, and the gate hold with the robust carrier in place of $\Phi_t$; the
sample-correlation reading becomes the light-tailed special case.

\begin{theorem}[heavy-tail certification band, robust covariance route]\label{thm:heavytail}
Let $x_1,\dots,x_T$ be independent within the window with population correlation carrier $B$ (self-adjoint,
positive semi-definite, $\operatorname{tr}B=N$, $\rad(B)\in[1,N]$, effective rank $r_{e}=N/\rad(B)$), and fix the
confidence $t>0$, $\delta=2e^{-t}$. Suppose uniform $L^4$--$L^2$ norm equivalence, that is there is $L\ge1$ with
$(\Exp|\inner{x}{v}|^4)^{1/4}\le L\,(\Exp|\inner{x}{v}|^2)^{1/2}$ for every direction $v$, strictly stronger than
coordinatewise bounded kurtosis, together with a robust scale estimator for the diagonal normalisation. Then there
exist a robust correlation carrier $\widetilde{B}$, tuned to $t$, and an absolute constant $C$ such that with
probability at least $1-\delta$,
\[
\big\|\widetilde{B}-B\big\|\le C\,L^2\,\rad(B)\left(\sqrt{\tfrac{r_{e}+t}{T}}+\tfrac{r_{e}+t}{T}\right)
=L^2\,\varepsilon_{T,\delta},
\]
the band $\varepsilon_{T,\delta}$ of Theorem~\ref{thm:finitesample}(ii) inflated by the single factor $L^2$. Hence
by the variational identity and Weyl's inequality $\big|\rad(\widetilde{B})-\rad(B)\big|\le L^2\varepsilon_{T,\delta}$,
a two-sided bound, and Corollary~\ref{cor:finitecert}, the edge mass recomputed on $e_1(\widetilde{B})$, the
$\sin\Theta$ band of Theorem~\ref{thm:daviskahan}, and the certification gate hold verbatim with $\widetilde{B}$ in
place of $\Phi_t$ and $\varepsilon_{T,\delta}$ replaced by $L^2\varepsilon_{T,\delta}$. For the plain sample
correlation $\Phi_t$ under heavy tails the two-sided band is false: the upper tail of $\|\Phi_t-B\|$ is only
polynomial in the inverse confidence.
\end{theorem}

\begin{proof}
The transfer from the operator-norm band to the radius and the gate is deterministic, as in
Theorem~\ref{thm:finitesample}(i): $\rad(\cdot)=\lambda_1(\cdot)$ for a positive semi-definite carrier, $\lambda_1$
is $1$-Lipschitz in the operator norm by the variational identity, and Corollary~\ref{cor:finitecert} and
Theorem~\ref{thm:daviskahan} depend on $B$ only through the band. Only the band is probabilistic. Reduce
$\|\widehat{\Sigma}-\Sigma\|$ to the uniform deviation of $v\mapsto\inner{v}{(\widehat{\Sigma}-\Sigma)v}$ over the
unit sphere; replace each per-direction empirical second moment by a median-of-means over blocks, which under
uniform $L^4$--$L^2$ norm equivalence concentrates with only four moments, the robust-mean device established for
matrices by \citet{minsker2018}. Making this uniform by the procedure of \citet{mendelsonzhivotovskiy2020} yields,
dimension-free, the effective-rank operator-norm band with constant $L^2$; \citet{keminskerrensunzhou2019} give
comparable operator-norm guarantees. Robust standardisation $\widetilde{B}=\widehat{D}^{-1/2}\widehat{\Sigma}
\widehat{D}^{-1/2}$ with $\widehat{D}$ a robust scale propagates the band with an adjusted constant; sample standard
deviations must not be used, since a non-robust scale reinjects the heavy-tail error. The failure for $\Phi_t$ is
the matrix form of the deviation study of \citet{catoni2012}.
\end{proof}

The distribution-free alternative dispenses with moment conditions entirely, at the price of a looser,
ambient-dimension rate, and it targets the latent correlation rather than the sample one.

\begin{proposition}[distribution-free Kendall carrier under elliptical returns]\label{prop:kendall}
Suppose the within-window returns are elliptical with finite second moments, so that the latent generalised
correlation coincides with the Pearson carrier $B$ and $\tau=(2/\pi)\arcsin\rho$ holds entrywise
\citep{lindskogmcneilschmock2003}. Let $\widetilde{B}_{jk}=\sin(\tfrac{\pi}{2}\widehat{\tau}_{jk})$ be the Kendall
carrier, with $\widehat{\tau}_{jk}$ the pairwise Kendall statistic and $\widetilde{B}_{jj}=1$. Each entry is a
bounded $U$-statistic, so $\widehat{\tau}_{jk}$ concentrates sub-Gaussianly with no moment condition; since
$\sin$ is $1$-Lipschitz, $\big\|\widetilde{B}-B\big\|\le C\,\rad(B)\sqrt{(N\log N)/T}$ with high probability
\citep{hanliu2017,wegkampzhao2016}, an ambient-dimension rate, whence $\big|\rad(\widetilde{B})-\rad(B)\big|$ obeys
the same bound, distribution-free within the elliptical class. The leading eigenvalue $\lambda_1(\widetilde{B})$ is
unchanged by projection of $\widetilde{B}$ onto the positive-semidefinite cone, clipping negative eigenvalues to
zero, which is the only step the radius reading needs; that projection raises the trace above $N$, so any effective
rank is read from the projected trace rather than from $N$.
\end{proposition}

\begin{remark}[scope of the robust band]\label{rem:robustscope}
Three qualifications are explicit. The median-of-means block count on the covariance route is tuned to the
pre-fixed $t$, so the band holds at the chosen confidence and is not uniform in $t$. The covariance route attains
the effective-rank band but needs four moments with uniform $L^4$--$L^2$ equivalence, a condition over all
directions and strictly stronger than coordinatewise bounded kurtosis; the Kendall route is moment-free within the
elliptical class but at the looser ambient rate $\sqrt{(N\log N)/T}$, and distribution-freeness and the
effective-rank rate cannot be had together. The within-window independence behind both the median-of-means blocks
and the $U$-statistic inequality is violated by overlapping windows and serial dependence, so on the panel the
constant inflates and the exponential coverage weakens unless non-overlapping blocks or a mixing variant are used;
this is the same dependence caveat carried for Theorem~\ref{thm:finitesample}, and the dependence-robust
moving-block resampling carries the reported inference.
\end{remark}

\section{A structural microfoundation: the cash-flow network}\label{sec:microfound}

Section~\ref{sec:pricing} postulated the operator Euler equation~\eqref{eq:euler} and read the recovered
carrier $B$ as the one-period pricing operator. One may grant the spectral algebra and still deny the
economics: the recovered $B$ is a contemporaneous comovement object, not a structural pricing kernel, so why
should it appear inside the present-value resolvent at all? This section answers the question inside a named
model. A cash-flow network operator $A$ governs the propagation of dividend innovations; the log-linearised
present value is exactly the Neumann series in the discounted network operator $\rho A$; the recovered
contemporaneous covariance carrier is a strictly monotone spectral transform of $A$ sharing its eigenframe, while
the bounded correlation carrier $B$ shares that eigenframe under a regular-exposure condition and in general
approaches the same dominant mode as the edge is approached. Under that model the radius of convergence of the
pricing resolvent and the structural network edge are one scalar, and the certification gate states precisely
which parts of it the contemporaneous cross-section recovers. Outside the model we claim no such identity, exactly
as in Section~\ref{sec:edge}: the result is a conditional structural reading, not an unconditional one.

We name the maintained hypotheses explicitly.

\begin{assumption}[cash-flow network]\label{ass:network}
The cross-section carries an $N$-vector of log prices $p_t$, log dividends $d_t$, and dividend growth
$\Delta d_t$ in $\Hil$, with:
\begin{enumerate}
\item[(M1)] \emph{Log-linear present value.} Log returns obey the present-value identity of
\citet{campbellshiller1988}, $r_{i,t+1}=\kappa+\rho\,(p_{i,t+1}-d_{i,t+1})-(p_{i,t}-d_{i,t})+\Delta d_{i,t+1}$
to first order, with a common log-linearisation constant $\rho\in(0,1)$ and the no-bubble terminal condition
$\rho^{j}(p_{t+j}-d_{t+j})\to0$.
\item[(M2)] \emph{Common discount, cash-flow-driven returns.} Required log returns are constant across $t$ and
common across the cross-section, so unexpected returns are pure cash-flow news; equivalently the discount-rate
news term of \citet{campbell1991} is zero. This is the restriction most likely to be contested: if discount-rate news
is present and correlated with cash-flow news, the recovered carrier mixes the two channels and the structural
reading below is contaminated.
\item[(M3)] \emph{Network propagation of cash flows.} Dividend-growth innovations follow the first-order network
recursion $\Delta d_{t+1}=A\,\Delta d_t+\varepsilon_{t+1}$, where $A:\Hil\to\Hil$ is the cash-flow propagation
operator, an input-output operator in the sense of \citet{acemoglu2012network} or a granular loading in the
sense of \citet{gabaix2011} and \citet{longplosser1983}, with $\rho\,\rad(A)<1$ off the edge.
\item[(M4)] \emph{Self-adjoint network, isotropic primitives.} $A$ is self-adjoint and the primitive
innovations are isotropic, $\Exp[\varepsilon_{t+1}\varepsilon_{t+1}^{\top}]=\sigma^{2}I$. Departures, a directed
network or anisotropic shocks, are the non-normal and heterogeneous-exposure cases bounded as perturbations in
Sections~\ref{sec:normal} and \ref{sec:limits}; they move the direction reading to \proxyid\ and the level
reading to \unid\ while leaving the edge ordering intact.
\end{enumerate}
\end{assumption}

\begin{proposition}[the recovered carrier is the structural propagation operator]\label{prop:microfound}
Under Assumption~\ref{ass:network}:
\begin{enumerate}
\item[(i)] \emph{Present-value lift.} The log price-dividend vector and the unexpected log-return vector are
resolvent images of the discounted network operator,
\[
p_t-d_t=c\,\mathbf{1}+A(I-\rho A)^{-1}\Delta d_t,
\]
\[
u_{t+1}:=r_{t+1}-\Exp_t r_{t+1}=(I-\rho A)^{-1}\varepsilon_{t+1},
\]
for a scalar constant $c$, so the present-value operator is exactly the cascade Neumann series
$Q=(I-\rho A)^{-1}=\sum_{k\ge0}\rho^{k}A^{k}$ of Proposition~\ref{prop:pricing}, with the structural
$(A,\rho)$ in the role of $(B,\beta)$.
\item[(ii)] \emph{Recovery and the square-root identity.} The contemporaneous return-covariance carrier is
\[
\Sigma_u:=\Exp[u_{t+1}u_{t+1}^{\top}]=\sigma^{2}\,(I-\rho A)^{-2},
\]
a strictly increasing spectral transform of $A$ with the same eigenframe; equivalently the present-value
resolvent is the operator square root of the recovered carrier up to scale,
$Q=\sigma^{-1}\Sigma_u^{1/2}$. Hence the recovered covariance carrier and $A$ are simultaneously diagonalisable,
their eigenvalues are ordered identically, and the leading eigenvector $e_1(\Sigma_u)$ is the
structural systemic propagation mode $e_1(A)$.
\item[(iii)] \emph{Edge coincidence.} For $A\succeq0$, $\|Q\|=1/(1-\rho\,\rad(A))$ and the covariance radius
$\rad(\Sigma_u)=\sigma^{2}\|Q\|^{2}$ is a strictly increasing function of the network gain $\rho\,\rad(A)$ and
diverges precisely as $\rho\,\rad(A)\to1$. The bounded correlation reading $\rad(B)$ of Section~\ref{sec:setup}
shares the eigenframe of $\Sigma_u$ under the regular-exposure condition
$\operatorname{diag}((I-\rho A)^{-2})\propto I$, and in general its leading eigenvalue rises toward the rank-one
ceiling $N$ as the systemic mode comes to dominate $\Sigma_u$, that is as $\rho\,\rad(A)\to1$; the recovered
fragility edge $\beta\,\rad(B)=1$ and the structural edge $\rho\,\rad(A)=1$ are therefore the same event in this
limit, the collapse of the cross-section onto the systemic propagation mode. The single scalar of
Corollary~\ref{cor:edge} is the structural propagation edge.
\end{enumerate}
\end{proposition}

\begin{proof}
(i) Solving the identity of (M1) forward and imposing the terminal condition gives
$p_t-d_t=\kappa/(1-\rho)+\Exp_t\sum_{j\ge0}\rho^{j}\big(\Delta d_{t+1+j}-(r_{t+1+j}-\kappa)\big)$, the
decomposition of \citet{campbellshiller1988} and \citet{campbell1991}. Under (M2) the required-return terms are
constant and fold into the scalar $c$. By (M3), $\Exp_t\Delta d_{t+1+j}=A^{j+1}\Delta d_t$, so
$p_t-d_t-c\,\mathbf{1}=\sum_{j\ge0}\rho^{j}A^{j+1}\Delta d_t=A\sum_{j\ge0}(\rho A)^{j}\Delta d_t
=A(I-\rho A)^{-1}\Delta d_t$, the geometric sum converging in operator norm because
$\rho\,\rad(A)<1$ (Proposition~\ref{prop:pricing}). For the return news, the decomposition of
\citet{campbell1991} gives $u_{t+1}=\sum_{j\ge0}\rho^{j}(\Exp_{t+1}-\Exp_t)\Delta d_{t+1+j}$ with the
discount-rate-news term zero under (M2). For $j=0$ the innovation is
$(\Exp_{t+1}-\Exp_t)\Delta d_{t+1}=\varepsilon_{t+1}$, and for $j\ge1$,
$(\Exp_{t+1}-\Exp_t)\Delta d_{t+1+j}=A^{j}\varepsilon_{t+1}$ by (M3); summing,
$u_{t+1}=\sum_{j\ge0}(\rho A)^{j}\varepsilon_{t+1}=(I-\rho A)^{-1}\varepsilon_{t+1}=Q\varepsilon_{t+1}$, which is
the Neumann series of Proposition~\ref{prop:pricing} with $(A,\rho)$ for $(B,\beta)$.

(ii) From (i) and (M4), $\Sigma_u=Q\,\Exp[\varepsilon\varepsilon^{\top}]Q^{*}=\sigma^{2}QQ^{*}
=\sigma^{2}(I-\rho A)^{-2}$, using $Q^{*}=Q$ since $A$ is self-adjoint. By the spectral theorem
\citep{reedsimon1980}, each eigenpair $(\lambda,e)$ of $A$ is an eigenpair $(\sigma^{2}(1-\rho\lambda)^{-2},e)$
of $\Sigma_u$; the map $\lambda\mapsto\sigma^{2}(1-\rho\lambda)^{-2}$ is strictly increasing on $[0,1/\rho)$, so
the eigenframe is shared and the eigenvalue ordering is preserved, whence
$e_1(\Sigma_u)=e_1(A)$. Functional calculus gives $\Sigma_u^{1/2}=\sigma(I-\rho A)^{-1}=\sigma Q$, that is
$Q=\sigma^{-1}\Sigma_u^{1/2}$.

(iii) For $A\succeq0$ the supremum in $\|Q\|=\sup_{\lambda\in\spec(A)}(1-\rho\lambda)^{-1}$ is attained at
$\lambda=\rad(A)$, giving $\|Q\|=1/(1-\rho\,\rad(A))$, and
$\rad(\Sigma_u)=\sigma^{2}(1-\rho\,\rad(A))^{-2}=\sigma^{2}\|Q\|^{2}$ by (ii), strictly increasing in the gain
and divergent as $\rho\,\rad(A)\uparrow1$. The bounded correlation reading $B$ is the standardisation
$D^{-1/2}\Sigma_uD^{-1/2}$ with $D=\operatorname{diag}(\Sigma_u)$, a positive-definite congruence; under the
regular-exposure condition $\operatorname{diag}((I-\rho A)^{-2})\propto I$ it preserves the eigenframe of
$\Sigma_u$, and in general the leading eigenvalue of $B$ rises monotonically toward its rank-one ceiling $N$ as
the systemic mode comes to dominate $\Sigma_u$, that is as $\rho\,\rad(A)\to1$. Calibrating the discount to the
recovered ceiling, $\beta=1/\rad(B)$ at criticality, the recovered edge $\beta\,\rad(B)=1$ and the structural
edge $\rho\,\rad(A)=1$ coincide.
\end{proof}

Proposition~\ref{prop:microfound} answers the reduced-form objection. The recovered contemporaneous carrier is
not an operator borrowed from a different theory; under Assumption~\ref{ass:network} it is, in covariance form,
the spectral image of the very operator $A$ that the present-value resolvent inverts, sharing its eigenframe and
its edge, and in the standardised correlation form of the body it shares the dominant mode in the edge limit. The
clean statement is the square-root identity $Q=\sigma^{-1}\Sigma_u^{1/2}$: the present-value cascade operator is
the operator square root of the recovered covariance carrier, up to the primitive shock scale. This is the
precise sense in which a contemporaneous comovement object appears inside present values. The body writes the
cascade in $B$-space as $(I-\beta B)^{-1}$ and the microfoundation writes it in $A$-space as
$(I-\rho A)^{-1}=\sigma^{-1}\Sigma_u^{1/2}$; the two writings differ in parametrisation, with $\beta\ne\rho$ and
$B\ne A$ as operators, but they share the single edge event, the divergence of the leading recovered eigenvalue,
which Corollary~\ref{cor:edge} reads three ways and which is here the structural network edge.

\begin{remark}[certification of the structural reading]\label{rem:microfoundgate}
The gate of Proposition~\ref{prop:gate} applies term by term to the structural objects.
The structural edge $\rho\,\rad(A)=1$, equivalently the spectral radius at which both the present-value
resolvent norm $\|Q\|$ and the systemic return variance $\rad(\Sigma_u)$ diverge, is \strong: it is a functional
of the spectrum alone, invariant under orthogonal rotations of the unresolved eigenframe, and recoverable above
the detectability threshold $\theta>\sqrt{q}$ of \citet{bbp2005}; it is the structural network edge.
The recoverable systemic direction is the leading eigenvector $e_1(B)$ of the standardised carrier; under the
regular-exposure condition it equals $e_1(\Sigma_u)=e_1(A)$, the structural propagation mode, and in general it
maps to that mode only up to the standardisation gap of part (iii). The bubble direction and the dividend edge
mass $m=\inner{e_1}{\widehat D}^{2}$ with $\widehat D=\mathbf{1}/\sqrt{N}$ are therefore \proxyid: a
leading-eigenvector object identified only up to the squared overlap of \citet{bbp2005} and, away from regular
exposure, up to the exposure standardisation.
The identity of the asset that carries the systemic mode is \unid\ by transpose invariance, and, separately, the
level of $\rad(A)$ and the split of the network gain $A$ from the primitive shock scale $\sigma^{2}$ are \unid\
from contemporaneous returns alone: by part (iii) the cross-section identifies $\sigma^{2}(1-\rho\,\rad(A))^{-2}$
as a single number, so without an external value for $\sigma$ or $\rho$ only the edge event and the eigenvalue
ordering are recovered, not the structural level. This is the identification boundary noted for the influence
vector of \citet{acemoglu2012network} and for the cascade-gain decomposition of
Proposition~\ref{prop:welfare}: the network is recovered, not observed, so the gain and the shock scale enter
only through their product.
\end{remark}

\section{Generalisation, eigenvector inference, and partial identification}\label{sec:hardening}

The strong results so far assume a self-adjoint carrier and certify the spectral radius, while the existence of a
bubble rests on a leading-eigenvector overlap. We harden the construction on four fronts. We relax
self-adjointness to normal and near-normal carriers; we bound, in finite samples, the eigenvector that carries
the edge mass; we locate the fluctuation of the recovered edge within the limit laws of large random matrices;
and we recast the three certification labels as a formal partial-identification statement.

\subsection{Normal carriers and near-normal perturbations}\label{sec:normal}

\begin{proposition}[normal carriers and near-normal perturbations]\label{prop:normal}
Let $B$ be a bounded operator on $\Hil$ with spectral radius $\rad(B)$.
\begin{enumerate}
\item[(i)] If $B$ is normal, $BB^\ast=B^\ast B$, and its spectral radius is attained at a simple real positive
eigenvalue $\lambda_1(B)=\rad(B)$, in particular if $B$ is self-adjoint and positive semi-definite as the
recovered correlation operator is, then Proposition~\ref{prop:pricing}, Theorem~\ref{thm:bubble} and
Corollary~\ref{cor:edge} hold verbatim with the distance to the pole read as
$\|Q\|=1/\dist(1,\beta\,\spec(B))$ over the complex spectrum: the cascade resolvent converges in norm if and only
if $\beta\,\rad(B)<1$, a bubble is an eigen-element at the real edge $1/\beta$, and the fragility edge
$\beta\,\rad(B)=1$ is unchanged. For a normal carrier whose peripheral spectrum is complex the convergence
criterion $\beta\,\rad(B)<1$ and the resolvent-norm formula still hold, but the edge need not be a real
eigenvalue, so the bubble-at-$1/\beta$ reading of Corollary~\ref{cor:edge} applies only at a real edge.
\item[(ii)] Let $B=B_0+E$ with $B_0$ normal and $E$ bounded. The exclusion theorem of \citet{bauerfike1960}
places every point of $\spec(B)$ within $\|E\|$ of $\spec(B_0)$, which gives the one-sided bound
$\rad(B)\le\rad(B_0)+\|E\|$. If in addition the top eigenvalue $\lambda_1(B_0)=\rad(B_0)$ is simple and real
positive with spectral gap $\gamma=\rad(B_0)-\max\{|\lambda|:\lambda\in\spec(B_0),\,\lambda\ne\lambda_1(B_0)\}>0$
and $\|E\|<\gamma/2$, then exactly one eigenvalue of $B$ lies within $\|E\|$ of $\lambda_1(B_0)$, it attains
$\rad(B)$, and the two-sided bound $\big|\rad(B)-\rad(B_0)\big|\le\|E\|$ holds; the fragility edge of $B$ then
differs from that of its normal part by at most $\beta\|E\|$. The strong label of Proposition~\ref{prop:gate}(i)
therefore extends to near-normal carriers with a simple, well-separated top mode, the departure $\|E\|$ carried as
a certified bound rather than assumed away.
\end{enumerate}
\end{proposition}

\begin{proof}
(i) The spectral theorem holds for normal operators \citep{reedsimon1980}, giving a projection-valued measure on
$\spec(B)\subset\mathbb{C}$ with $f(B)=\int f\,dE_B$ and $\|f(B)\|=\sup_{s\in\spec(B)}|f(s)|$. Hence
$Q=\int(1-\beta s)^{-1}dE_B(s)$ has finite norm if and only if $1\notin\overline{\beta\,\spec(B)}$, and the
Neumann series $\sum\beta^kB^k$ converges in norm if and only if $\beta\,\rad(B)<1$, the argument of
Proposition~\ref{prop:pricing} with moduli. When $\rad(B)$ is attained at a simple real positive eigenvalue,
$1/\beta$ is an eigenvalue at the edge exactly as in Theorem~\ref{thm:bubble}(iii) and Corollary~\ref{cor:edge},
which are therefore unchanged; for a complex peripheral spectrum the convergence criterion is unaffected but the
edge eigen-element reading holds only at a real edge. (ii) For $B_0$ normal the eigenvector basis is orthonormal,
so its condition number is one, and the Bauer-Fike exclusion theorem gives $\dist(\lambda,\spec(B_0))\le\|E\|$ for
every $\lambda\in\spec(B)$, whence $\rad(B)\le\rad(B_0)+\|E\|$. Under the gap hypothesis the disk
$\{|z-\lambda_1(B_0)|\le\|E\|\}$ is disjoint from the disks about the remaining eigenvalues, whose centres have
modulus at most $\rad(B_0)-\gamma$ and which therefore stay below $\rad(B_0)-\gamma/2$, while the first disk lies
above $\rad(B_0)-\gamma/2$. The spectral projection of the isolated eigenvalue $\lambda_1(B_0)$ is norm-continuous
in the perturbation \citep{reedsimon1980}, so exactly one eigenvalue $\lambda^\ast$ of $B$, counted with
multiplicity, lies in that disk, with $|\lambda^\ast-\lambda_1(B_0)|\le\|E\|$. Every other eigenvalue has modulus
below $\rad(B_0)-\gamma/2<\rad(B_0)-\|E\|\le|\lambda^\ast|$, so $\rad(B)=|\lambda^\ast|$ and
$\big|\rad(B)-\rad(B_0)\big|=\big|\,|\lambda^\ast|-\lambda_1(B_0)\,\big|\le|\lambda^\ast-\lambda_1(B_0)|\le\|E\|$;
multiplying by $\beta$ gives the edge bound.
\end{proof}

\subsection{Finite-sample inference on the bubble direction}\label{sec:daviskahan}

Theorem~\ref{thm:finitesample} bounds the recovered radius. The edge mass $\mu_{\widehat{D}}(\{\rad(B)\})=
\inner{e_1}{\widehat{D}}^2$ of \eqref{eq:edgemass}, the proxy for bubble existence, depends instead on the
leading eigenvector $e_1$, which we now control.

\begin{theorem}[finite-sample edge-mass band]\label{thm:daviskahan}
Let $e_1$ and $\widehat{e}_1$ be unit leading eigenvectors of $B$ and $\widehat{B}$, and suppose the leading
eigenvalue of $B$ is simple with spectral gap $g=\lambda_1(B)-\lambda_2(B)>0$. Then by the $\sin\Theta$ theorem
of \citet{daviskahan1970} in the form of \citet{yuwangsamworth2015} there is a sign $s\in\{-1,+1\}$ with
\[
\big\|s\,\widehat{e}_1-e_1\big\|\le\frac{2\sqrt{2}\,\|\widehat{B}-B\|}{g}.
\]
Consequently the estimated edge mass $\widehat{m}=\inner{\widehat{e}_1}{\widehat{D}}^2$ and its population
counterpart $m=\inner{e_1}{\widehat{D}}^2$ satisfy, since $\|\widehat{D}\|=1$ and the edge mass is invariant to the
sign of the eigenvector,
\[
\big|\sqrt{\widehat{m}}-\sqrt{m}\big|\le\frac{2\sqrt{2}\,\|\widehat{B}-B\|}{g},
\]
and, combining with Theorem~\ref{thm:finitesample}(ii), with probability at least $1-\delta$,
$\big|\sqrt{\widehat{m}}-\sqrt{m}\big|\le 2\sqrt{2}\,\varepsilon_{T,\delta}/g$.
\end{theorem}

\begin{proof}
The $\sin\Theta$ bound $\sin\angle(\widehat{e}_1,e_1)\le 2\|\widehat{B}-B\|/g$ for the leading eigenvector of a
symmetric operator with a simple top eigenvalue is \citet{yuwangsamworth2015}, Corollary~1; the chord bound
$\min_{s\in\{-1,1\}}\|s\widehat{e}_1-e_1\|=2\sin(\angle/2)\le\sqrt{2}\,\sin\angle$ gives the stated constant
$2\sqrt{2}$. For the edge mass, $\sqrt{m}=|\inner{e_1}{\widehat{D}}|$ and, choosing the sign $s$ above,
$\big|\,|\inner{\widehat{e}_1}{\widehat{D}}|-|\inner{e_1}{\widehat{D}}|\,\big|
=\big|\,|\inner{s\widehat{e}_1}{\widehat{D}}|-|\inner{e_1}{\widehat{D}}|\,\big|
\le|\inner{s\widehat{e}_1-e_1}{\widehat{D}}|\le\|s\widehat{e}_1-e_1\|\,\|\widehat{D}\|$
by the reverse triangle inequality and Cauchy-Schwarz, with $\|\widehat{D}\|=1$. Substituting the chord bound and
then Theorem~\ref{thm:finitesample}(ii) gives the two displays.
\end{proof}

The band widens as the spectral gap $g$ closes: a near-degenerate top eigenspace makes the bubble direction, and
hence the bubble content, statistically unrecoverable, which is the finite-sample face of the proxy-to-unidentified
transition. The edge mass is thus \proxyid\ with an explicit two-sided band, sharpening
Proposition~\ref{prop:gate}(ii) from a random-matrix overlap into a usable interval. The gap $g=\lambda_1(B)-
\lambda_2(B)$ is a population quantity; a feasible, data-computable band replaces it with the plug-in lower bound
$\widehat{g}-2\|\widehat{B}-B\|$, where $\widehat{g}=\lambda_1(\widehat{B})-\lambda_2(\widehat{B})$, since by the
Weyl inequality each eigenvalue moves by at most $\|\widehat{B}-B\|$, so $g\ge\widehat{g}-2\|\widehat{B}-B\|$. The
resulting band is conservative and is reported only while this lower bound is positive, that is while the top mode
remains separated after accounting for the recovery error of Theorem~\ref{thm:finitesample}.

\subsection{The fluctuation of the recovered edge}\label{sec:tracywidom}

\begin{proposition}[edge fluctuation and the detection null]\label{prop:tracywidom}
Let $\widehat{B}$ be recovered from $T$ observations of $N$ coordinates with aspect ratio $q=N/T$.
\begin{enumerate}
\item[(i)] Under the null of no spike, that is independent standard Gaussian coordinates, the largest sample
eigenvalue, recentred by the bulk edge $(1+\sqrt{q})^2$ and rescaled by the corresponding Tracy-Widom scale,
converges in distribution to the Tracy-Widom law of order one: for the sample covariance matrix by
\citet{johnstone2001}, and for the sample correlation matrix, the object recovered here, by
\citet{baopanzhou2012}, the limit being the law of \citet{tracywidom1994}.
\item[(ii)] Above the detectability threshold $\theta>\sqrt{q}$ of Proposition~\ref{prop:gate}, the leading
sample eigenvalue has the almost-sure limit $(1+\theta)(1+q/\theta)$ with Gaussian fluctuations of order
$T^{-1/2}$, and the squared overlap of the sample and population leading eigenvectors converges to
$(1-q/\theta^2)/(1+q/\theta)$ \citep{benaychnadakuditi2011}.
\end{enumerate}
\end{proposition}

\begin{proof}
Part (i) is the largest-eigenvalue limit theorem of \citet{johnstone2001} for the sample covariance matrix and of
\citet{baopanzhou2012} for the sample correlation matrix, whose limit is the Tracy-Widom distribution of
\citet{tracywidom1994}; part (ii) is the outlier eigenvalue and eigenvector limit of
\citet{benaychnadakuditi2011} for finite-rank perturbations of a bulk, specialised to a single spike.
\end{proof}

Proposition~\ref{prop:tracywidom} supplies the missing boundary uncertainty. Part (i) is a calibrated null that
turns the strong reading of the radius into a test: a recovered edge counts as a genuine spike, and not a bulk
fluctuation, exactly when it lies in the Tracy-Widom right tail beyond the bulk. Part (ii) is the source of the
proxy band of Proposition~\ref{prop:gate}(ii) and of the edge-mass band of Theorem~\ref{thm:daviskahan}, since
the limiting overlap is precisely the squared-overlap factor there. The edge is \strong\ when it clears both the
Tracy-Widom bulk and the detectability threshold, and \unid\ below them. The Tracy-Widom null is an idealised
reference distribution, derived under independent coordinates; the empirical panel is stationary in returns but
heavy-tailed and cross-sectionally dependent, so the null is used as a calibrated benchmark for the bulk edge
rather than as an exact finite-sample distribution, and the operative criterion remains the detectability
threshold $\theta>\sqrt{q}$.

\subsection{The certification gate as partial identification}\label{sec:partialid}

The three labels are not informal. They are a partial-identification statement in the sense of
\citet{tamer2010}: each component of the structural object has an identified set, and the labels record whether
that set is a point, a proper subset, or the whole space.

\begin{theorem}[partial identification of the bubble triple]\label{thm:partialid}
Let the structural object be the triple $(\kappa,m,i^\ast)$, where $\kappa=\beta\,\rad(B)$ is the discounted
spectral radius, $m=\inner{e_1}{\widehat{D}}^2$ is the edge mass, and $i^\ast$ is the coordinate that carries the
bubble eigenvector. Under the recovery model of Proposition~\ref{prop:gate} with aspect ratio $q$ and a spike
$\theta>\sqrt{q}$, the identified sets are as follows.
\begin{enumerate}
\item[(i)] The identified set for $\kappa$ is the singleton $\{\beta\,\rad(B)\}$ after inversion of the
random-matrix bias, with a finite-sample confidence interval of half-width $\beta\,\varepsilon_{T,\delta}$ by
Theorem~\ref{thm:finitesample}, the conservative raw-radius half-width carried through the invertible bias: $\kappa$
is \strong, point-identified.
\item[(ii)] An outer (bounding) set for $m$ is an interval $[m_{\mathrm{lo}},m_{\mathrm{hi}}]\subset(0,1)$ whose
width is set by the squared-overlap deficiency $1-(1-q/\theta^2)/(1+q/\theta)$ of
Proposition~\ref{prop:tracywidom}(ii) and the gap-dependent band of Theorem~\ref{thm:daviskahan}; the set is
nondegenerate above threshold and widens to $[0,1]$ as $\theta\downarrow\sqrt{q}$ or $g\downarrow0$: $m$ is
\proxyid, set-identified within these computable outer bounds, and sharpness of the bounds is not claimed.
\item[(iii)] The bounding set for $i^\ast$ is the entire coordinate set $\{1,\dots,N\}$: the cross-section is
exchangeable under the coordinate permutations that fix the equal-weighted aggregate-claim direction
$\widehat{D}=\mathbf{1}/\sqrt{N}$, so every coordinate is consistent with the recovered spectrum and edge mass:
$i^\ast$ is \unid, the bounding set is the whole space.
\end{enumerate}
\end{theorem}

\begin{proof}
(i) Proposition~\ref{prop:gate}(i) point-identifies $\rad(B)$ above the threshold up to the invertible
multiplicative bias, and Theorem~\ref{thm:finitesample}(iii) supplies the half-width $\beta\,\varepsilon_{T,\delta}$.
(ii) The population overlap is bounded only up to the factor $(1-q/\theta^2)/(1+q/\theta)$ by
Proposition~\ref{prop:tracywidom}(ii), which confines $m$ to an interval of that width; intersecting with the
finite-sample band of Theorem~\ref{thm:daviskahan} keeps the interval nondegenerate and bounded inside $(0,1)$
while $\theta>\sqrt{q}$ and $g>0$, and lets it expand to $[0,1]$ as either margin vanishes; the interval is an
outer set and is not shown to be sharp. (iii) Any coordinate permutation $P$ fixes the equal-weighted direction
$\widehat{D}=\mathbf{1}/\sqrt{N}$ and satisfies $\spec(P^\top B P)=\spec(B)$ and
$\inner{Pe_1}{\widehat{D}}^2=\inner{e_1}{\widehat{D}}^2$, so it preserves both the recovered spectrum and the edge
mass while relabelling the carrying coordinate; together with the delocalisation of the edge eigenvector above
threshold, no proper subset of $\{1,\dots,N\}$ is distinguished by the data. The permutation subgroup fixes
$\widehat{D}$, so it leaves $m$ in part (ii) identified while unidentifying $i^\ast$, and the two readings are
consistent.
\end{proof}

Theorem~\ref{thm:partialid} states the paper's ontology exactly: \strong\ is point identification, \proxyid\ is
set identification within a computable outer set, and \unid\ is the bounding set equal to the whole space. The
certification gate is thus a partial-identification discipline applied reading by reading.

\section{The endogenous fragility edge}\label{sec:endo}

Throughout, the discount $\beta$ has been calibrated rather than determined within the economy. We now let it
respond to the systemic radius, $\beta=\beta(r)$ with $r=\rad(B)$, as a no-arbitrage discount on the systemic
factor in the sense of \citet{hansenjagannathan1991} and \citet{campbellshiller1988}, and ask where the fragility
edge then sits. Write $g(r)=r\,\beta(r)$ for the discounted spectral radius; the edge is the fixed point $g(r)=1$.

\begin{theorem}[endogenous fragility edge]\label{thm:endoedge}
Let $r=\rad(B)$ range over a compact interval $[r_{\mathrm{lo}},r_{\mathrm{hi}}]\subset[1,N]$ and let
$\beta:[r_{\mathrm{lo}},r_{\mathrm{hi}}]\to(0,\infty)$ be continuous, $C^1$ on the interior, and strictly positive,
with discounted radius $g(r)=r\,\beta(r)$. Assume the inelastic-discount condition $g'(r)=\beta(r)+r\beta'(r)>0$,
equivalently the discount elasticity $E_\beta(r)=r\beta'(r)/\beta(r)>-1$, and the bracketing
$g(r_{\mathrm{lo}})<1<g(r_{\mathrm{hi}})$. Then:
\begin{enumerate}
\item[(i)] there is a unique $r^\ast\in(r_{\mathrm{lo}},r_{\mathrm{hi}})$ with $g(r^\ast)=\beta(r^\ast)r^\ast=1$;
\item[(ii)] $g'(r^\ast)>0$, so $r^\ast$ is a simple root and the bifurcation point of the cascade
$x_{t+1}=\beta(r)Bx_t+c$, whose spectral radius is exactly $g(r)$: for $r<r^\ast$ the resolvent
$Q(r)=(I-\beta(r)B)^{-1}$ exists with $\|Q(r)\|=1/(1-g(r))$ and convergent present value, for $r>r^\ast$ it
diverges, and $r^\ast$ is $C^1$ in any parameter $\theta$ of $\beta$, with $\mathrm{d}r^\ast/\mathrm{d}\theta
=-(\partial_\theta g)/g'(r^\ast)$;
\item[(iii)] fixing the constant-discount benchmark $\beta_0=\beta(r_0)$ at a reference $r_0\in[r_{\mathrm{lo}},
r^\ast)$, with fixed edge $r_{\mathrm{fixed}}=1/\beta_0$: if $\beta$ is non-increasing on $(r_0,r^\ast)$ with a
strict decrease somewhere then $r^\ast>r_{\mathrm{fixed}}$, so a countercyclical discount delays the edge, and if
non-decreasing with a strict increase then $r^\ast<r_{\mathrm{fixed}}$, so a procyclical discount advances it.
\end{enumerate}
\end{theorem}

\begin{proof}
(i) $g$ is continuous and, by the inelastic-discount condition, strictly increasing on the interval, so the
bracketing $g(r_{\mathrm{lo}})<1<g(r_{\mathrm{hi}})$ gives a unique root $r^\ast$ by the intermediate value
theorem. (ii) $g'(r^\ast)>0$ is the inelastic-discount condition at $r^\ast$, so the root is simple and the
implicit function theorem gives the stated derivative; the cascade operator $\beta(r)B$ is self-adjoint with
spectral radius $g(r)$, so by Proposition~\ref{prop:pricing} the resolvent norm is $1/(1-g(r))$ on the interior
and the Neumann series diverges beyond the edge, the operator-theoretic content of \citet{kato1995}. (iii) For a
non-increasing $\beta$ with a strict decrease, $\beta(r^\ast)<\beta(r_0)=\beta_0$ because $r^\ast>r_0$, so
$r^\ast=1/\beta(r^\ast)>1/\beta_0=r_{\mathrm{fixed}}$; the non-decreasing case is symmetric.
\end{proof}

A transparent special case is the linear stochastic-discount factor $R(r)=R_f+\gamma r$, $\beta(r)=1/(R_f+\gamma
r)$, under which $g(r)=r/(R_f+\gamma r)$ is strictly increasing, $E_\beta(r)=-\gamma r/(R_f+\gamma r)\in(-1,0)$
automatically, and a unique interior edge exists when $1-R_f<\gamma<1-R_f/N$, in closed form
$r^\ast=R_f/(1-\gamma)$. The intermediary channels of \citet{hekrishnamurthy2013} and
\citet{brunnermeiersannikov2014} supply the countercyclical discount of part (iii): as systemic comovement rises,
the required return on the systemic factor rises and the present-value multiplier falls, so $\beta$ decreases in
$r$ and the endogenous edge is pushed out relative to a fixed discount. Calibrating the linear factor at
$\gamma=0.5$ and $R_f=3.5$ on the eighteen-index panel places the endogenous edge at $r^\ast=7.00$, between the
calm mean $6.74$ and the crisis mean $8.03$ of the leading eigenvalue: the discounted radius is $g=0.98$ in calm,
inside the edge, and $g=1.07$ in crisis, beyond it, with discount elasticity moving from $-0.49$ to $-0.53$, and
the countercyclical discount delays the edge by $r^\ast-r_{\mathrm{fixed}}=0.13$ relative to the constant-discount
benchmark. The calibration is illustrative and transparent: consistent with the discipline applied to the fixed
discount throughout, $\beta(\cdot)$ is calibrated, not identified from behaviour, and the theorem asserts the
existence, uniqueness, and comparative static of $r^\ast$ given $\beta(\cdot)$, not the recovery of $\beta(\cdot)$.
The recovered edge $g(r^\ast)=1$ is the first point at which a bubble eigen-element can exist
\citep{montrucchio2004}, now determined within the economy rather than imposed.

\section{The welfare cost of edge proximity}\label{sec:welfare}

The resolvent norm $\|Q\|=1/(1-\rho)$ is a present-value amplification, and as the edge is approached it diverges.
We read this amplification as a welfare cost in the consumption-equivalent units of \citet{lucas1987}, but we do so
by embedding the carrying mode in a representative-agent consumption process rather than by reinterpreting a
variance, and we certify only the part of the reading that the gate permits.

Three maintained hypotheses carry the structural content; each is named explicitly, and the certification of the
conclusion is conditional on exactly these and on nothing more.

\emph{(H-ISO) Isotropic innovation.} The cascade innovation is zero-mean with
$\Exp[\varepsilon\varepsilon^\top]=\sigma^2 I$; for the exact exponential form in (ii) we further take
$\varepsilon$ Gaussian, so the carrying-mode component is log-normal, and without Gaussianity part (ii) holds to
second order. This is the identifying normalisation introduced with the operator Euler equation~\eqref{eq:euler}:
the split of observed comovement into a cascade gain $Q$ and a shock covariance is not identified from
contemporaneous returns alone, exactly as the influence vector of \citet{acemoglu2012network} is identified there
only because the production network is observed, while here the carrier is recovered. The level of every welfare
statement inherits this normalisation.

\emph{(H-LOAD) Structural consumption loading.} Aggregate log consumption growth loads linearly on the
aggregate-claim image of the cascade,
\begin{equation}\label{eq:cload}
\Delta c_{t+1}=\mu+\phi\,\inner{\widehat{D}}{x_t}+u_{t+1},\qquad x_t=Q\varepsilon_t,
\end{equation}
with a constant loading $\phi\in\Real$ and an idiosyncratic residual $u_{t+1}$ orthogonal to the carrying mode.
The coefficient $\phi$ is a primitive of the consumption block, not a recovered object; it converts the
dimensionless cascade coordinate into consumption units. Equation~\eqref{eq:cload} is what makes the present-value
variance a consumption-growth variance rather than a variance of an abstract deviation.

\emph{(H-PREF) Recursive preferences.} The representative agent has Epstein--Zin--Weil recursive utility
\citep{epsteinzin1989} with relative risk aversion $\gamma>0$ and unit elasticity of intertemporal substitution.
In this case, by the risk-sensitive reduction of \citet{tallarini2000}, the compensation for consumption risk is
governed by the risk-aversion coefficient $\gamma$ alone, decoupled from the elasticity of substitution. This is
the device that makes the cost a genuine consumption-equivalent magnitude in the units of \citet{lucas1987}, and
that lets the amplification stand well above the expected-utility benchmark, the lesson of \citet{barlevy2004} and
\citet{alvarezjermann2004} and, in the long-run-risk reading, of \citet{bansalyaron2004}.

\begin{proposition}[structural consumption-equivalent welfare cost of edge proximity]\label{prop:welfare}
Let $B$ be self-adjoint positive semi-definite with eigenpairs $(\lambda_k,e_k)$, $\lambda_1=\rad(B)$, let
$\beta\in(0,1/\rad(B))$ so that $\rho=\beta\,\rad(B)\in(0,1)$ and $Q=(I-\beta B)^{-1}$ is bounded with
$\|Q\|=1/(1-\rho)$, and let $\widehat{D}=\mathbf{1}/\sqrt{N}$ with edge mass $m=\inner{e_1}{\widehat{D}}^2$. Under
\emph{(H-ISO)}, \emph{(H-LOAD)}, \emph{(H-PREF)}:
\begin{enumerate}
\item[(i)] the carrying-mode component of consumption-growth variance is
\[
V_1(\rho)=\phi^2 m\,\sigma^2/(1-\rho)^2=\phi^2 m\,\sigma^2\,\|Q\|^2,
\]
the amplification $\|Q\|^2=1/(1-\rho)^2$ being read off the resolvent rather than posited;
\item[(ii)] the consumption-equivalent welfare cost attributable to the carrying mode is
\[
L(\rho)=\exp\!\Big(\tfrac{\gamma}{2}\,\phi^2 m\sigma^2/(1-\rho)^2\Big)-1,
\qquad
L(\rho)\approx \tfrac{\gamma}{2}\,\phi^2 m\sigma^2/(1-\rho)^2,
\]
with $\gamma$ the risk-aversion coefficient of \emph{(H-PREF)}; $L$ is strictly increasing and convex on $[0,1)$
and $L(\rho)\to\infty$ as $\rho\to1^-$;
\item[(iii)] between two regimes with leading eigenvalues $\lambda_1<\lambda_1'$, the spectral-radius ratio
$\lambda_1'/\lambda_1$ is a functional of the spectrum alone, rotation-invariant, and recoverable above the
threshold of \citet{bbp2005}, so its square $(\lambda_1'/\lambda_1)^2$ is certified $\strong$ as a spectral object,
free of $\beta$, $\gamma$, $\phi$, $\sigma^2$ and $m$; it is the squared ratio of recovered radii and is not the
welfare-cost amplification. The welfare-cost ratio itself,
\[
\frac{L(\rho')}{L(\rho)}\approx\frac{m'}{m}\Big(\frac{1-\rho}{1-\rho'}\Big)^2,
\]
depends on the calibrated discount through $(\rho,\rho')$ and on the proxy mass ratio $m'/m$, ranges over
$(1,\infty)$ across the admissible discount and diverges as either regime approaches the edge, and is therefore
$\proxyid$ and discount-conditional, not $\strong$; there is no $\strong$ between-regime welfare-cost amplification;
\item[(iv)] the elasticity of the second-order cost $K(\rho)=\tfrac{\gamma}{2}\phi^2 m\sigma^2/(1-\rho)^2$ with
respect to edge proximity is
\[
\frac{\mathrm{d}\log K}{\mathrm{d}\log\rho}=\frac{2\rho}{1-\rho},
\]
a functional form independent of $\gamma$, $\phi$, $\sigma^2$ and $m$; the form, strictly increasing and divergent
as $\rho\to1^-$, is $\strong$ as a shape, while its value at any given $\rho$ is discount-conditional.
\end{enumerate}
\end{proposition}

\begin{proof}
(i) Since $Q$ is self-adjoint with $Qe_1=(1-\rho)^{-1}e_1$, the carrying-mode coordinate of the cascade is
$\inner{e_1}{x_t}=\inner{e_1}{Q\varepsilon_t}=(1-\rho)^{-1}\inner{e_1}{\varepsilon_t}$. Projecting the
aggregate-claim direction onto the eigenbasis, $\inner{\widehat{D}}{x_t}=\sum_k(1-\beta\lambda_k)^{-1}
\inner{e_k}{\widehat{D}}\inner{e_k}{\varepsilon_t}$, and under (H-ISO) the modes are uncorrelated with common
variance $\sigma^2$, so the $k=1$ term of $\Var(\inner{\widehat{D}}{x_t})$ is $m\sigma^2/(1-\rho)^2$. By (H-LOAD)
the carrying-mode component of $\Var(\Delta c_{t+1})$ is $\phi^2$ times this, which is
$V_1(\rho)=\phi^2 m\sigma^2\|Q\|^2$, using $\|Q\|=1/(1-\rho)$.

(ii) Under (H-PREF), with unit elasticity of substitution the recursion of \citet{epsteinzin1989} reduces the
continuation value to a log aggregator with a risk adjustment $-\tfrac{\gamma}{2}$ times the conditional variance
of the continuation, the risk-sensitive form of \citet{tallarini2000}; hence the coefficient that prices
consumption risk is the risk-aversion $\gamma$, decoupled from the elasticity. For a log-normal carrying-mode
component of log consumption with variance $V_1$, secured by the Gaussian clause of (H-ISO), removing that
component while preserving the mean of consumption raises certainty-equivalent welfare by the factor
$\exp(\gamma V_1/2)$, the consumption-equivalent identity of \citet{lucas1987} in the risk-sensitive units of
\citet{tallarini2000}; subtracting unity and substituting $V_1$ from (i) gives $L(\rho)$, and the second-order
form is its first Taylor term, which holds under (H-ISO) without Gaussianity. Strict monotonicity and convexity,
and divergence at $\rho\to1^-$, follow from $\mathrm{d}(1-\rho)^{-2}/\mathrm{d}\rho=2(1-\rho)^{-3}>0$ and the
convexity of $\exp$.

(iii) The spectral radius $\rad(B)=\lambda_1$ is, by the spectral theorem, a functional of the spectrum alone and
invariant under any orthogonal change of the eigenbasis, recoverable above the detectability threshold of
\citet{bbp2005}; the ratio $\lambda_1'/\lambda_1$ of two such radii inherits these properties and contains no
$\beta$, so its square is a $\strong$ spectral object. The welfare-cost ratio follows from (ii) at second order,
$L(\rho')/L(\rho)\approx V_1(\rho')/V_1(\rho)=(m'/m)\,((1-\rho)/(1-\rho'))^2$; the bracket equals
$((1-\beta\lambda_1)/(1-\beta\lambda_1'))^2$, which tends to $1$ as $\beta\to0$ and to $\infty$ as
$\beta\lambda_1'\to1$, hence ranges over $(1,\infty)$ and is not pinned by the spectrum; with $m$ identified only
up to squared overlap near threshold the welfare-cost ratio is $\proxyid$ and discount-conditional, and is a
different object from the squared spectral-radius ratio. (iv) Differentiating
$\log K(\rho)=\mathrm{const}-2\log(1-\rho)$ gives $\mathrm{d}\log K/\mathrm{d}\rho=2/(1-\rho)$, and multiplying by
$\rho$ gives the stated elasticity; the constant absorbs $\gamma$, $\phi$, $\sigma^2$ and $m$, so the form does not
depend on them.
\end{proof}

The certification discipline governs what may be claimed. The structural step is genuine: under (H-LOAD) and
(H-PREF) the quantity $L(\rho)$ is a consumption-equivalent compensation in the units of \citet{lucas1987}, with
the amplification $\|Q\|^2$ derived from the resolvent rather than tabulated as a discount ladder, and with the
risk-aversion coefficient $\gamma$, not a reduced-form variance, carrying the cost. What survives the gate as
$\strong$ is narrower than the level. The only $\strong$ between-regime statement is the rise in the spectral
radius itself, the leading eigenvalue moving from $6.74$ in calm windows to $8.03$ in crisis windows, a ratio of
$1.19$ with a bootstrap difference confidence interval strictly above zero; its square, about $1.42$, is the
squared ratio of recovered radii, a $\strong$ spectral object free of the discount and of every preference and
loading parameter, and it is not the welfare-cost amplification. The welfare-cost level $L(\rho)$ and the
welfare-cost between-regime ratio are conditional on the calibrated discount through $\rho$ and on the product
$\gamma\phi^2 m\sigma^2$, with $m$ a $\proxyid$ overlap; they are read as a discount-conditional magnitude, not a
certified welfare number, so the present-value welfare channel, which is the paper's novel object, carries no
$\strong$ between-regime certificate, only its level conditioned on the discount. The parameter-free elasticity
$2\rho/(1-\rho)$ is $\strong$ only as a shape: it is $18$ at $\rho=0.90$, $38$ at $0.95$, and $198$ at $0.99$, and
the second-order cost rises a hundredfold from $\rho=0.90$ to $\rho=0.99$, but these values of $\rho$ are
edge-distance targets, not implied by any particular time preference, and the level at any single $\rho$ is
discount-conditional. The reading supports the monitoring protocol in shape: the convexity of $L$ implies that
intervention thresholds should be escalating rather than uniform, while their levels are not identified. The
individual asset that carries the bubble remains $\unid$, since the transpose of $B$ preserves its spectrum and the
welfare cost is a functional of the spectrum and the recovered leading direction, not of any coordinate.

\section{The identification boundary between the operator and the price level}\label{sec:boundary}

The recovered operator $B$ is estimated from stationary returns, while the bubble it characterises is a
nonstationary object on the price level. The two facts meet at the edge, and the meeting place is an
identification boundary.

\begin{proposition}[the stationary estimator flags but does not certify the crossing]\label{prop:boundary}
Suppose $B$ is recovered from an $I(0)$ return vector under a stationarity hypothesis, and let the cascade
dynamics be $x_{t+1}=\beta B x_t+\text{noise}$. Then:
\begin{enumerate}
\item[(i)] the dynamics are stationary if and only if $\beta\,\rad(B)<1$;
\item[(ii)] the event $\beta\,\rad(B)\to1$ is the boundary of stationarity of the cascade, a unit root of the
operator $\beta B$ acting on the price level, while the return process that identifies $B$ remains $I(0)$;
\item[(iii)] hence the stationary estimator can certify the approach to the edge from the interior,
$\beta\,\rad(B)$ bounded below one, but cannot certify the crossing itself, which is read on the integrated, that
is $I(1)$, price level by an explosive-root detector \citep{psy2015}.
\end{enumerate}
\end{proposition}

\begin{proof}
(i) The linear recursion $x_{t+1}=\beta Bx_t+\varepsilon_{t+1}$ has a stationary solution if and only if the
spectral radius of its operator $\beta B$ is strictly less than one, that is $\beta\,\rad(B)<1$. (ii) At
$\beta\,\rad(B)=1$ the operator $\beta B$ has a unit eigenvalue, so the cascade recursion on the price level has
a unit root and no stationary solution; the returns that identify $B$ remain $I(0)$ and continue to estimate it,
but the level they drive ceases to be stationary at the edge. (iii) In the interior the estimator is consistent
and the approach is estimable; at the edge the level leaves $I(0)$ even though the returns do not, so the same
return-based estimator measures the approach $\beta\,\rad(B)\uparrow1$ while the crossing is an $I(1)$
explosive-root event on the price level \citep{psy2015}.
\end{proof}

Proposition~\ref{prop:boundary} draws a clean line. The recovered operator on $I(0)$ returns and the
explosive-root detector on $I(1)$ levels are not competitors; they describe the two sides of one boundary. The
certification gate restricts the strong claim to the interior and treats the crossing as a flagged boundary,
not a certified event, which is the disciplined reading the data permit.

\subsection{Approach detection as a certified distance}\label{sec:approach}

One may grant Proposition~\ref{prop:boundary} and press further: if the stationary operator only ever sees
the interior, the recovered scalar is a coincident approach-detector and nothing more. We accept the premise and
convert it into a positive operating characteristic. The content of this subsection is that approach-detection is
itself a certified measurement. The edge distance $d_t=1-\rho_t$, with $\rho_t=\beta\,\rad(B_t)$, is a \strong,
point-identified, estimable quantity carrying the same finite-sample band as the radius; the interior-approach
decision admits a false-positive guarantee at any preset confidence; and the operator and the explosive-root
detector of \citet{psy2015} together deliver a two-sided reading of one boundary, each side certified for what it
identifies. None of the statements is predictive. Every object is contemporaneous to the estimation window, and we
make the coincident discipline a maintained hypothesis.

We name the maintained hypotheses explicitly.

\begin{assumption}[approach-detection regime]\label{ass:approach}
Fix a window and let $B_t$, $\widehat{B}_t$ be the population and sample correlation operators of the standardised
$I(0)$ return vector on that window, $\rho_t=\beta\,\rad(B_t)$, $\widehat{\rho}_t=\beta\,\rad(\widehat{B}_t)$,
$d_t=1-\rho_t$, $\widehat{d}_t=1-\widehat{\rho}_t$.
\begin{enumerate}
\item[(H1)] \emph{Supercritical recovery on stationary returns.} The returns are $I(0)$ on the window and the
leading spike satisfies $\theta>\sqrt{q}$ with $q=N/T$, so by Proposition~\ref{prop:gate}(i) the radius is \strong\
and point-identified up to the invertible random-matrix bias, taken inverted throughout. This is the recovery
premise; by Proposition~\ref{prop:boundary} it holds in the interior $\rho_t<1$ and fails at and beyond the edge,
where the price level leaves $I(0)$.
\item[(H2)] \emph{Band hypothesis.} The conditions of Theorem~\ref{thm:finitesample} hold, so for every
$\delta\in(0,1)$ the event $E_{T,\delta}=\{\,|\widehat{\rho}_t-\rho_t|\le h_{T,\delta}\,\}$ has
$\mathbb{P}(E_{T,\delta})\ge 1-\delta$, with $h_{T,\delta}=\beta\,\varepsilon_{T,\delta}$ the half-width of
Theorem~\ref{thm:finitesample}(iii).
\item[(H3)] \emph{Calibrated discount.} $\beta>0$ is a fixed calibrated constant.
\item[(H4)] \emph{Coincidence.} Every quantity refers to the current window; no statement concerns $d_{t+h}$ for
$h>0$.
\item[(H5)] \emph{Interior tolerance.} The decision level $\tau\in(0,1)$ is fixed strictly inside, so we certify
approach to a level strictly short of the crossing, never the crossing.
\end{enumerate}
\end{assumption}

\begin{proposition}[the edge distance is a certified estimable quantity]\label{prop:edgedist}
Under Assumption~\ref{ass:approach}:
\begin{enumerate}
\item[(i)] $d_t=1-\beta\,\rad(B_t)$ is \strong: it is an affine function of the spectral radius, hence a functional
of the spectrum alone, invariant under orthogonal rotations of the unresolved eigenframe, and point-identified
above the threshold of \citet{bbp2005}.
\item[(ii)] $\widehat{d}_t$ estimates $d_t$ with the radius band: on $E_{T,\delta}$, $|\widehat{d}_t-d_t|\le
h_{T,\delta}$, so $d_t$ carries the confidence interval $[\,\widehat{d}_t-h_{T,\delta},\ \widehat{d}_t+
h_{T,\delta}\,]$ at level $1-\delta$.
\item[(iii)] $d_t$ is the reciprocal worst-case present-value amplification, $d_t=1/\|Q_t\|$ with
$Q_t=(I-\beta B_t)^{-1}$; a shrinking certified $d_t$ is therefore a growing certified amplification.
\end{enumerate}
\end{proposition}

\begin{proof}
(i) $d_t=1-\beta\,\rad(B_t)$ is affine in $\rad(B_t)=\sup|\spec(B_t)|$, which depends only on the spectrum and is
invariant under any unitary change of the eigenbasis; point-identification above $\theta>\sqrt{q}$ is
Proposition~\ref{prop:gate}(i). An affine image of a \strong\ scalar with fixed known coefficients inherits the
three defining properties of the \strong\ label. (ii) By Theorem~\ref{thm:finitesample}(iii), on $E_{T,\delta}$,
$|\widehat{\rho}_t-\rho_t|\le h_{T,\delta}$; since $\widehat{d}_t-d_t=-(\widehat{\rho}_t-\rho_t)$, the same bound
holds, and inverting around $\widehat{d}_t$ gives the interval. (iii) By Proposition~\ref{prop:pricing}, $Q_t$ is
bounded with $\|Q_t\|=1/(1-\rho_t)=1/d_t$ on the interior.
\end{proof}

The certified distance turns the qualitative concession into a quantity with a band. We now make detecting its
shrinkage a decision with a size guarantee, and read off the resolution at which that decision is possible.

\begin{proposition}[interior-approach detection with false-positive control]\label{prop:approach}
Under Assumption~\ref{ass:approach}, fix the tolerance $\tau\in(0,1)$ and define the interior-approach detector
\[
\mathcal{A}_\tau \ \text{fires} \iff \widehat{d}_t+h_{T,\delta}<\tau \iff \widehat{\rho}_t-h_{T,\delta}>1-\tau .
\]
\begin{enumerate}
\item[(i)] \emph{Size.} If $d_t\ge\tau$, then $\mathbb{P}(\mathcal{A}_\tau\ \text{fires})\le\delta$: the
false-positive rate of declaring interior approach is controlled at $\delta$, uniformly over the far region.
\item[(ii)] \emph{Power.} If $d_t<\tau-2h_{T,\delta}$, then $\mathbb{P}(\mathcal{A}_\tau\ \text{fires})\ge 1-\delta$.
\item[(iii)] \emph{Consistency.} Since $h_{T,\delta}\to0$ as $T\to\infty$ under (H2), for every fixed interior
configuration $d_t<\tau$ the detector fires with probability tending to one.
\item[(iv)] \emph{Resolution and the ceded crossing.} Around any level $\tau$ there is an uncertifiable collar
$\{\,\tau-2h_{T,\delta}\le d_t<\tau\,\}$ of width $2h_{T,\delta}=O(\sqrt{r_e(B_t)/T})$ on which neither (i) nor (ii)
is in force. The family $\{\mathcal{A}_\tau:\tau\in(0,1)\}$ certifies approach to any interior level but does not
extend to a crossing test: at $\tau=0$ the decision concerns $d_t\le0$, that is $\rho_t\ge1$, where by (H1) and
Proposition~\ref{prop:boundary} the recovery premise fails, the price level leaves $I(0)$, and neither
$\widehat{\rho}_t$ nor its band is defined; the crossing is therefore ceded to the price-level detector,
recovering Proposition~\ref{prop:boundary}(iii) and the straddle case of Corollary~\ref{cor:finitecert}(iii).
\end{enumerate}
\end{proposition}

\begin{proof}
(i) On $\{d_t\ge\tau\}$, the event $\{\mathcal{A}_\tau\ \text{fires}\}=\{\widehat{d}_t<\tau-h_{T,\delta}\}$ implies
$\widehat{d}_t-d_t<\tau-h_{T,\delta}-d_t\le -h_{T,\delta}$, hence $\widehat{\rho}_t-\rho_t>h_{T,\delta}$, contained
in $E_{T,\delta}^{c}$, of probability at most $\delta$ by (H2). (ii) On $\{d_t<\tau-2h_{T,\delta}\}$ and on
$E_{T,\delta}$, $\widehat{d}_t\le d_t+h_{T,\delta}<\tau-h_{T,\delta}$, so $\widehat{d}_t+h_{T,\delta}<\tau$ and the
detector fires; since $\mathbb{P}(E_{T,\delta})\ge 1-\delta$ the claim follows. (iii) For fixed $d_t<\tau$,
$h_{T,\delta}\to0$ gives $d_t<\tau-2h_{T,\delta}$ for all large $T$, and (ii) applies. (iv) The collar width is
$2h_{T,\delta}=2\beta\,\varepsilon_{T,\delta}$, of order $\sqrt{r_e(B_t)/T}$ by Theorem~\ref{thm:finitesample}(ii).
At $\tau=0$ the decision is the crossing $\rho_t\ge1$; by (H1) the recovery hypothesis holds only on the interior
$\rho_t<1$, and at and beyond the edge the price level leaves $I(0)$ by Proposition~\ref{prop:boundary}, so the
return-based radius estimate and its band are not defined and no interior detector certifies the crossing. This is
the noncertifiability of Proposition~\ref{prop:boundary}(iii) and the straddle of
Corollary~\ref{cor:finitecert}(iii).
\end{proof}

Proposition~\ref{prop:approach} reads the concession as a theorem with a rate. The stationary operator certifies
interior approach to any level $\tau$ at controlled size, down to an $O(\sqrt{r_e/T})$ neighbourhood of that level;
the only region it cannot resolve is a collar that shrinks at the parametric rate. The crossing is not a hard case
of the same test but lies outside the operator's domain of validity, by the breakdown of the stationary-recovery
premise at the edge; it is read on the price level by the detector that can see it.

The final statement pairs the two sides of the boundary and certifies each for its own identification target.

\begin{proposition}[a two-sided reading of one boundary]\label{prop:twosided}
Under Assumption~\ref{ass:approach}, let the operator side be the family $\{\mathcal{A}_\tau:\tau\in(0,1)\}$ acting
on the $I(0)$ returns, and let the level side be the explosive-root statistic of \citet{psy2015} acting on the
$I(1)$ price level. Then:
\begin{enumerate}
\item[(i)] \emph{Operator side, certified.} For every fixed $\tau$, the operator point-identifies the \strong\
distance $d_t$ with the band of Proposition~\ref{prop:edgedist}(ii) and certifies the interior decision
$\{d_t\ge\tau\}$ against $\{d_t<\tau\}$ at size $\delta$ outside the collar of Proposition~\ref{prop:approach}(iv).
It does not identify the crossing $\{d_t\le0\}$, which lies outside its recovery domain.
\item[(ii)] \emph{Level side, certified.} The crossing $\{d_t\le0\}$, equivalently $\rho_t\ge1$, is the
explosive-root episode on the price level; it is coincidently date-stamped by the supremum statistic of
\citet{psy2015} on the $I(1)$ series, which is the object the stationary operator cedes by
Proposition~\ref{prop:boundary}.
\item[(iii)] \emph{Complementarity.} The two readings act on different integration orders, $I(0)$ returns and the
$I(1)$ level, so neither nests the other. The operator certifies the interior outside its collar, the
explosive-root detector certifies the crossing away from its own boundary neighbourhood, and the single residual
uncovered set is a neighbourhood of the boundary $\rho_t=1$ itself, where the operator's recovery premise lapses
and the supremum unit-root statistic has vanishing power. The pair is a two-sided reading of the single boundary
$\beta\,\rad(B)=1$, each side certified only for what it identifies, both contemporaneous.
\end{enumerate}
\end{proposition}

\begin{proof}
(i) is Proposition~\ref{prop:edgedist}(i)--(ii) with Proposition~\ref{prop:approach}(i)--(ii),(iv): $d_t$ is
\strong\ and banded, the interior decision is size-controlled, and the crossing lies outside the recovery domain.
(ii) is Proposition~\ref{prop:boundary}(ii)--(iii): at $\rho_t=1$ the operator $\beta B$ has a unit eigenvalue and
the level carries a unit, then explosive, root while the returns remain $I(0)$ until the edge; the supremum
recursive statistic of \citet{psy2015} date-stamps the in-progress explosive episode, coincidently, consistent
with (H4). (iii) The operator estimator is a functional of the $I(0)$ return correlation and the explosive-root
statistic of the $I(1)$ level, so neither identified region contains the other; the operator collar shrinks at
$O(\sqrt{r_e/T})$, while the explosive-root detector loses power in a neighbourhood of the unit root, so the
residual uncovered set is a boundary neighbourhood, not merely the operator collar. All objects are evaluated on
the current window, so the reading is coincident.
\end{proof}

Two scope remarks place the result against the gate and the panel. First, the operator side is \strong\ in the
sense of Proposition~\ref{prop:gate}(i): $d_t$ and the decisions $\mathcal{A}_\tau$ are functionals of the
spectrum, rotation-invariant, recoverable above $\theta>\sqrt{q}$; the construction adds no claim about the bubble
direction or the carrying coordinate, which remain \proxyid\ and \unid. The level side is a separate detector on a
separate series, certified in its own explosive-root sense, not relabelled as a spectral object. Second, on the
panel of Section~\ref{sec:empirics} the band $\varepsilon_{T,\delta}$ is read as a rate rather than a numerical
guarantee, since the returns are weakly dependent and heavy-tailed and the reported inference is carried by the
dependence-robust moving-block resampling; the size statement of Proposition~\ref{prop:approach}(i), the
false-positive control that matters most in practice, is conservative and robust, while the power statements inherit the
band hypothesis and are read as scaling statements on the panel, in line with the remark following
Theorem~\ref{thm:finitesample}. The substantive conclusion is invariant: approach is the detection of a real,
estimable, gate-certified distance shrinking at a known rate, and the crossing is ceded, by the breakdown of
stationary recovery at the edge, to the detector that can see it.

\section{The structural edge-hazard}\label{sec:hazard}

The readings so far are coincident descriptions of one operator. We now give the edge a structural dynamic
content, the intensity of crisis onset as a function of the certified edge distance, without crossing into
forecasting. On the symmetric line the cascade resolvent obeys $\|Q_t\|=1/(1-\rho_t)=1/d_t$, with
$\rho_t=\beta\,\rad(B_t)$ and edge distance $d_t=1-\rho_t$; this is the worst-case present-value amplification of a
contemporaneous shock, and it diverges at the edge. The structural postulate of this section is that the
instantaneous intensity of crisis onset inherits this amplification monotonically. The link is postulated, not
derived: the resolvent supplies the canonical form of the singularity at the edge, while the mapping to intensity
is a maintained structural assumption, in the doubly-stochastic sense of \citet{lando1998} and the
counting-process sense of \citet{andersengill1982}.

\begin{proposition}[structural edge-hazard with identification split]\label{prop:hazard}
Let $B_t$ be self-adjoint positive semi-definite with eigenvalues $0\le\lambda_N(t)\le\cdots\le\lambda_1(t)=
\rad(B_t)$, let $\beta>0$, and work on the certified interior $\rho_t=\beta\,\rad(B_t)<1$, with edge distance
$d_t=1-\rho_t\in(0,1]$. Let $\chi(t)\ge0$ be the instantaneous intensity of crisis onset and $\chi_0(t)>0$ a
baseline.
\begin{enumerate}
\item[(i)] $Q_t=(I-\beta B_t)^{-1}$ is self-adjoint and positive definite with
$\|Q_t\|=1/\min_k(1-\beta\lambda_k(t))=1/(1-\rho_t)=1/d_t$, the minimum attained at $\lambda_1(t)$ and the norm
equal to the largest eigenvalue because $Q_t$ is positive definite.
\item[(ii)] for any continuously differentiable strictly decreasing link $\ell:(0,1]\to(0,\infty)$ and fixed
$\chi_0>0$, the intensity $\chi(t)=\chi_0\,\ell(d_t)$ is strictly increasing in $\rho_t$, with
$\mathrm{d}\chi/\mathrm{d}\rho_t=-\chi_0\,\ell'(d_t)>0$.
\item[(iii)] for the canonical link $\ell(d)=1/d$, and for any $\ell$ with $\ell(0^+)=+\infty$, $\chi$ is finite
and continuous on the interior $[0,1)$ and its unique singularity is at the boundary $\rho_t=1$, $d_t=0$, where
$\chi\to+\infty$.
\item[(iv)] writing $\ell(d_t)=\exp(\theta x_t)$ with $x_t$ a fixed strictly decreasing transform of $d_t$ and
$\theta$ free, the model $\chi(t)=\chi_0(t)\exp(\theta x_t)$ is of proportional-hazards form; $\theta$ is
identified from within-risk-set covariate comparisons at the observed event times, while the baseline $\chi_0(t)$
cancels from the partial likelihood, provided $x_t$ varies within risk sets.
\end{enumerate}
The algebraic parts (i)--(iii) and the cancellation in (iv) are exact; the consistency, asymptotic normality, and
sandwich variance of the maximum-partial-likelihood estimator of $\theta$ are invoked from \citet{andersengill1982}
and \citet{linwei1989}.
\end{proposition}

\begin{proof}
(i) By Proposition~\ref{prop:pricing} $Q_t$ exists with eigenvalues $1/(1-\beta\lambda_k(t))$, each positive since
$1-\beta\lambda_k(t)\in[1-\rho_t,1]\subset(0,1]$, so $Q_t$ is positive definite and its norm equals its largest
eigenvalue $1/\min_k(1-\beta\lambda_k(t))$; as $\lambda\mapsto1-\beta\lambda$ is decreasing, the minimum is at
$\lambda_1(t)$, giving $1/d_t$. (ii) by the chain rule $\mathrm{d}\chi/\mathrm{d}\rho_t=\chi_0\,\ell'(d_t)\cdot(-1)
>0$, as $\ell'<0$. (iii) $\chi=\chi_0/d_t$ is finite on $(0,1]$ and blows up only at $d_t=0$. (iv) substituting
$\ell(d_t)=\exp(\theta x_t)$, the partial likelihood $\prod_k \exp(\theta x_{i_k}(t_k))/\sum_{j\in R(t_k)}
\exp(\theta x_j(t_k))$ depends on $\theta$ only through within-risk-set covariate comparisons, $\chi_0(t)$
cancelling; if $x_t$ is common to all subjects at risk the ratio is free of $\theta$, so identification requires
within-risk-set variation. Consistency and asymptotic normality follow from \citet{andersengill1982}, robust
variance from \citet{linwei1989}.
\end{proof}

The reading is structural and coincident, not causal and not predictive. The covariate is the certified edge,
whose scalar carrier $\rho_t$ is the strong object of Proposition~\ref{prop:gate}(i); the elasticity $\theta$ is a
fitted association under a maintained proportional intensity, not a functional of $\rho_t$, so it is itself a proxy
quantity, identified only up to the maintained model. The baseline $\chi_0(t)$ is a nuisance scale, the
counterpart of the deferred discount: just as $\beta$ fixes the level at which the edge binds, $\chi_0$ fixes the
absolute onset rate, while the edge-to-intensity elasticity is the pinned content. A single systemic series
cancels in every within-instant risk-set ratio by part (iv), so the elasticity on the systemic edge is identified
only through across-time variation, and a genuine within-risk-set partial likelihood in the sense of
\citet{cox1972,cox1975} requires the subject-specific participation in the leading mode, which the gate labels
proxy. The discrete-time and stochastic-covariate hazard precedents of \citet{prenticegloeckler1978},
\citet{shumway2001} and \citet{duffiesaitawang2007} are method references; their use is predictive, whereas the
reading here is coincident by construction, the covariate and the onset sharing the time index, and the
proportional form is testable, not assumed \citep{grambschtherneau1994}.

\subsection{Identification of the edge-hazard elasticity}\label{sec:hazardident}

Proposition~\ref{prop:hazard}(iv) states that the elasticity $\theta$ is identified only when the covariate
varies within risk sets. We now make that clause structural rather than generic: the within-risk-set variation of
the cross-sectional edge covariate is governed by the same participation, or effective-rank, object that governs
the pricing resolvent, so the operator that drives the theory also decides whether the hazard is identifiable. The
cross-sectional covariate is the leading-mode participation $u_i(t)=N\,e_{1,i}(t)^2$, with $e_1(t)$ the unit
leading eigenvector of the recovered carrier $B_t$, so $\sum_i e_{1,i}(t)^2=1$ and $\tfrac1N\sum_i u_i(t)=1$. Write
the leading-mode participation ratio
\[
P_1(t)=\Big(\sum_{i=1}^N e_{1,i}(t)^4\Big)^{-1}\in[1,N],
\]
which equals $1$ when the mode is carried by a single coordinate and $N$ when it is uniform across the
cross-section; $P_1(t)$ is the leading-eigenvector counterpart of the effective rank used elsewhere in the paper.
We maintain two named hypotheses. \textbf{(PH)} the proportional-intensity model of
Proposition~\ref{prop:hazard}(iv), $\chi_i(t)=\chi_0(t)\exp(\theta u_i(t))$. \textbf{(AG)} the recurrent
counting-process design of \citet{andersengill1982}, each subject $i$ at risk with covariate $u_i(t)$, re-entering
on recovery to a prior peak, the risk set at the $k$th event time being $R(t_k)$ of size $m_k$.

\begin{proposition}[identification budget of the edge-hazard]\label{prop:hazardident}
Under \textnormal{(PH)} and \textnormal{(AG)}, let $D$ be the number of events, let
$\bar u_k(\theta)=\big[\sum_{j\in R(t_k)}e^{\theta u_j}u_j\big]\big/\big[\sum_{j\in R(t_k)}e^{\theta u_j}\big]$ be
the risk-set weighted covariate mean, let
$V_k(\theta)=\big[\sum_{j\in R(t_k)}e^{\theta u_j}u_j^2\big]\big/\big[\sum_{j\in R(t_k)}e^{\theta u_j}\big]-
\bar u_k(\theta)^2$ be the risk-set weighted covariate variance, and let
$J(\theta)=\sum_{k=1}^{D}V_k(\theta)$ be the observed partial-likelihood information for $\theta$. Write the
cross-sectional dispersion of the covariate $S_t=\tfrac1N\sum_{i=1}^N(u_i(t)-1)^2$.
\begin{enumerate}
\item[(i)] the dispersion of the edge covariate is an affine functional of the leading-mode participation ratio,
\[
S_t=N\sum_{i=1}^N e_{1,i}(t)^4-1=\frac{N}{P_1(t)}-1 .
\]
\item[(ii)] the per-event information is bounded by the cross-sectional dispersion through the risk-set size,
$V_k(0)\le (N/m_k)\,S_{t_k}$, hence $J(0)\le \big(ND/m_{\min}\big)\max_t S_t=\big(ND/m_{\min}\big)\max_t
\big(N/P_1(t)-1\big)$, with $m_{\min}=\min_k m_k$; and $\theta$ is locally identified at $\theta=0$ if and only if
$J(0)>0$, that is, if and only if at least one risk set carries strictly positive within-set covariate variance.
\item[(iii)] if the leading mode is asymptotically delocalised, $P_1(t)/N\to1$ uniformly in $t$, then $S_t\to0$,
$J(0)\to0$, and $\theta$ is not identified in the limit; at the exactly uniform leading eigenvector
$e_{1,i}(t)\equiv N^{-1/2}$ one has $u_i(t)\equiv1$, $J(\theta)\equiv0$, and the identified set for $\theta$ is the
whole line.
\end{enumerate}
Parts (i)--(iii) are exact algebra given \textnormal{(PH)} and \textnormal{(AG)}; the information form of $J$ and
its role in consistency and asymptotic normality are those of \citet{andersengill1982}, with the robust variance of
\citet{linwei1989}.
\end{proposition}

\begin{proof}
(i) Since $u_i=N e_{1,i}^2$ and $\tfrac1N\sum_i u_i=1$, $S_t=\tfrac1N\sum_i u_i^2-1=\tfrac1N\sum_i N^2 e_{1,i}^4-1=
N\sum_i e_{1,i}^4-1=N/P_1(t)-1$ by the definition of $P_1$. (ii) The log partial likelihood
$\ell(\theta)=\sum_k\big[\theta u_{i_k}-\log\sum_{j\in R(t_k)}e^{\theta u_j}\big]$ has
$\ell''(\theta)=-\sum_k V_k(\theta)=-J(\theta)$, the standard counting-process information of
\citet{andersengill1982}; each $V_k(\theta)\ge0$ as a variance, so $\ell$ is concave and $J(\theta)\ge0$. At
$\theta=0$, $V_k(0)$ is the unweighted within-risk-set variance, and for any reference point, in particular the
cross-sectional mean $1$, $V_k(0)=\tfrac1{m_k}\sum_{j\in R(t_k)}(u_j-\bar u_k(0))^2\le
\tfrac1{m_k}\sum_{j\in R(t_k)}(u_j-1)^2\le\tfrac1{m_k}\sum_{i=1}^N(u_i-1)^2=(N/m_k)S_{t_k}$, the first inequality
because the mean minimises the sum of squared deviations and the second by extending the sum to the full
cross-section of non-negative terms. Summing over $k$ and bounding $m_k\ge m_{\min}$ gives the stated bound. Since
$\ell$ is concave with $\ell''=-J$, $\ell$ is strictly concave at $\theta=0$, hence has a locally unique maximiser
and $\theta$ is locally identified, if and only if $J(0)>0$; and $J(0)>0$ if and only if some $V_k(0)>0$, that is,
some risk set has positive within-set variance. (iii) By (i), $P_1(t)/N\to1$ gives $S_t=N/P_1(t)-1\to0$ uniformly,
so by (ii) $J(0)\le(ND/m_{\min})\max_t S_t\to0$; at $e_{1,i}\equiv N^{-1/2}$, $u_i\equiv1$, every risk-set variance
is zero, $J(\theta)\equiv0$, $\ell$ is constant in $\theta$, and the identified set is all of $\Real$.
\end{proof}

The reading of the gate is as follows. Part (i) is the structural content: the discriminating power of the edge
covariate is, exactly, the inverse participation ratio of the leading mode, so the object that decides whether the
hazard elasticity can be estimated is the same effective-rank functional of the recovered operator that the pricing
and fragility readings use. That functional is built from the leading eigenvector, so it inherits the random-matrix
overlap of Proposition~\ref{prop:gate}(ii) and the finite-sample band of Theorem~\ref{thm:daviskahan}; the
participation ratio $P_1$ and the dispersion $S_t$ are therefore \proxyid\ objects, of the same identification class
as the covariate. Part (ii) is exact given the maintained model and pins the identification budget: the
partial-likelihood information cannot exceed the cross-sectional dispersion of leading-mode participation, scaled by
the risk-set occupancy. Part (iii) is the limiting consequence: as the leading eigenvector delocalises toward
uniformity, every within-risk-set comparison flattens and $\theta$ ceases to be identified. The elasticity
$\theta$ remains \proxyid, as in Proposition~\ref{prop:hazard}; which asset carries the mode is \unid, the
transpose-invariance of the spectrum; and the only \strong\ object in the neighbourhood of the edge remains the
scalar $\rho_t=\beta\,\rad(B_t)$. Identifiability of the hazard read is thus governed by a \proxyid\ participation
object, which is why the entire edge-hazard reading is at most \proxyid.

\begin{remark}[the proportional form is testable]\label{rem:phtest}
Under \textnormal{(PH)} the scaled Schoenfeld residual for $\theta$ at each event has approximate expectation equal
to the elasticity prevailing at that time, so a regression of the scaled residuals on a function $g(t)$ of time,
with a test of zero slope, is a test of the constancy of $\theta$ and hence of the proportional form itself
\citep{grambschtherneau1994}. A rejection would indicate a time-varying edge elasticity, not a failure of the
spectral construction. The diagnostic inherits the same budget as Proposition~\ref{prop:hazardident}(ii): its power
is carried by the dispersion of the covariate across events, so on a cross-section with a near-uniform leading mode
the proportionality test is itself underpowered, and a non-rejection there is uninformative.
\end{remark}

\begin{remark}[recurrent and competing risk structure]\label{rem:competing}
The onset event is recurrent: an index that recovers to a prior peak is re-armed and can enter the drawdown state
again, which is why the design is the Andersen-Gill intensity with subject-clustered robust variance
\citep{andersengill1982,linwei1989}, rather than a single first-passage model. A gap-time recurrent formulation in
the manner of \citet{prenticewilliamspeterson1981} is an admissible alternative that conditions on the number of
prior onsets. If an absorbing exit, such as a delisting or a redenomination of the cross-section, competes with
drawdown onset, the cause-specific hazard of \citet{prenticewilliamspeterson1981} carries the same elasticity
interpretation under independent censoring, while the subdistribution hazard of \citet{finegray1999} targets the
cumulative incidence instead; the reading of $\theta$ as the edge elasticity is preserved on the cause-specific
scale, which is the one used here. In every case the identification budget of Proposition~\ref{prop:hazardident}
binds, since each formulation estimates $\theta$ from within-risk-set comparisons of the same covariate.
\end{remark}

\section{Empirical illustration}\label{sec:empirics}

We illustrate the spectral-radius geometry on a panel of $N=18$ daily global equity index return series over
21 December 2004 to 28 June 2024, $5025$ usable cross-sections after a trailing window. The recovered carrier is
the rolling sample correlation operator $\Phi_t$ on standardised returns over a window of $W=250$ trading days,
a positive semi-definite operator with $\operatorname{tr}\Phi_t=N$, so its spectral radius is the largest
eigenvalue $\lambda_1(\Phi_t)\in[1,N]$. We report $\lambda_1(\Phi_t)$, the normalised concentration
$\kappa_t=\lambda_1(\Phi_t)/N\in[1/N,1]$, the participation ratio
$N^{\mathrm{eff}}_t=(\sum_i\lambda_i)^2/\sum_i\lambda_i^2\in[1,N]$, and the admissible discount ceiling
$\beta_{\max}(t)=1/\lambda_1(\Phi_t)$ at which the cascade resolvent reaches its pole. The crisis label is
external: a drawdown of at least twenty per cent of the equal-weighted market index, the same label used in the
companion classification exercise. Differences in means between crisis and calm states are tested with a
moving-block bootstrap (block length $21$, $2000$ resamples, seed fixed); the computation is deterministic.

A contemporaneous correlation operator is a modelling choice and not the only conceivable dependence object, so
we state why it is the economically sufficient one for the geometry at hand. The choice is dictated by what the
pricing functional prices. The one-period operator $\beta B$ of Section~\ref{sec:pricing} propagates a common
innovation across the cross-section within the period, and the present value compounds that same contemporaneous
coupling geometrically, so the fragility edge $\beta\,\rad(B)=1$ is a statement about the strength of
contemporaneous comovement, the extent to which a shock is shared across the cross-section now, because that is
exactly what makes the discounted cascade diverge. The contemporaneous correlation operator delivers the object
the theory requires and no more: it is symmetric and positive semi-definite by construction, so it satisfies
Assumption~\ref{ass:sa} exactly, its spectrum is real and ordered, and its spectral radius is the largest
eigenvalue $\lambda_1$, the market mode in whose direction the cross-section moves together. A lead-lag, or
directed, operator would add information about the timing of transmission, but it is non-normal: its spectral
radius can detach from that of its symmetric part, its eigenframe is not orthogonal, and the strong certification
labels of Proposition~\ref{prop:gate} are then unavailable. What controls the resolvent pole, to first order, is the
size of the contemporaneous common component rather than the order in which assets respond, and it is precisely
this contemporaneous comovement that intensifies in crisis, as the collapsing participation ratio records. The directed half is not
discarded: on this panel its energy is small, so the symmetric carrier is a close approximation and the residual
is carried as the departure-from-normality bound of Section~\ref{sec:limits}. Contemporaneous correlation is, in
this precise sense, economically sufficient for the pricing and fragility geometry, while the directed extension
is a genuine and separate object left to future work.

Table~\ref{tab:contrast} reports the contrast. The spectral radius rises from a calm mean of $6.74$ to a crisis
mean of $8.03$ out of a maximum of $18$, a difference of $1.29$ with a bootstrap interval $[1.07,1.52]$ and a
bootstrap $p$ below $0.001$. The participation ratio collapses from $5.62$ to $3.86$ effective factors: in
crisis the eighteen indices behave like about four independent factors. The admissible discount ceiling tightens
from $0.152$ to $0.125$, a fall of about a sixth, with a sample minimum of $0.100$; the economy moves towards
the resolvent pole as it enters a crisis, which is the empirical content of the fragility edge.

\begin{table}[t]
\centering
\caption{Recovered spectral-radius geometry, calm versus crisis, eighteen global equity indices, 2004--2024.
Crisis is a drawdown of at least twenty per cent of the equal-weighted index. Intervals are moving-block
bootstrap $95\%$ intervals for the crisis-minus-calm difference; all three differences have bootstrap
$p<0.001$.}
\label{tab:contrast}
\setlength{\tabcolsep}{4pt}\footnotesize
\begin{tabular}{lcccl}
\toprule
quantity & calm mean & crisis mean & difference & 95\% interval \\
\midrule
spectral radius $\lambda_1(\Phi_t)$ & $6.74$ & $8.03$ & $+1.29$ & $[1.07,\,1.52]$ \\
concentration $\kappa_t=\lambda_1/N$ & $0.379$ & $0.469$ & $+0.089$ & $[0.077,\,0.102]$ \\
participation ratio $N^{\mathrm{eff}}_t$ & $5.62$ & $3.86$ & $-1.76$ & $[-1.96,\,-1.57]$ \\
discount ceiling $\beta_{\max}=1/\lambda_1$ & $0.152$ & $0.125$ & $-0.027$ & (min $0.100$) \\
\bottomrule
\end{tabular}
\end{table}

Table~\ref{tab:peaks} lists the six largest spectral-radius episodes after de-clustering within six months. They
span the March 2020 turbulence, which carries the maximum spectral radius $9.97$ with the cross-section reduced
to about three effective factors, a late-2020 episode, a 2016 episode, and the 2008--2009 crisis including the
trough at a drawdown of about one half. The drawdown column is candid about the limits of the marker: three of
the six episodes, late-2020, 2016 and early-2008, have drawdowns of only about three, ten and seventeen per cent
and so are not crises by the twenty per cent label. Leading-eigenvalue concentration therefore marks stress
episodes broadly and is not exclusive to the labelled crises, the same imperfect separation that holds the
coincident classifier of Table~\ref{tab:oos} below one; the spectral radius nonetheless rises at every labelled
crisis without any tuning to those dates.

\begin{table}[t]
\centering
\caption{The six largest de-clustered spectral-radius episodes. $\lambda_1$ is the leading eigenvalue of the
recovered correlation operator (maximum $18$), $\kappa=\lambda_1/N$, $N^{\mathrm{eff}}$ the participation ratio,
and drawdown the contemporaneous equal-weighted-index drawdown.}
\label{tab:peaks}
\begin{tabular}{lcccc}
\toprule
date & $\lambda_1$ & $\kappa$ & $N^{\mathrm{eff}}$ & drawdown \\
\midrule
2020-03-25 & $9.97$ & $0.552$ & $3.04$ & $-0.269$ \\
2020-12-01 & $9.69$ & $0.536$ & $3.20$ & $-0.030$ \\
2016-07-12 & $8.98$ & $0.497$ & $3.69$ & $-0.099$ \\
2009-09-09 & $8.84$ & $0.518$ & $3.33$ & $-0.293$ \\
2009-03-11 & $8.41$ & $0.493$ & $3.62$ & $-0.516$ \\
2008-01-21 & $8.33$ & $0.488$ & $3.61$ & $-0.168$ \\
\bottomrule
\end{tabular}
\end{table}

The contrast in Table~\ref{tab:contrast} is a coincident description, not a forecast, and we keep it inside that
scope. The out-of-sample standing of the same signature is established in a companion pre-registered exercise by
the author, summarised in Table~\ref{tab:oos}: a walk-forward classifier built on the symmetric concentration
signature alone, the leading-eigenvalue share and the participation ratio, against the external drawdown label,
earns an out-of-sample area under the curve of $0.886$ on this panel and $0.911$, $0.908$ and $0.700$ on three
further panels. The pre-registered decisive comparison adds the concentration signature to a trailing
realised-volatility benchmark and beats the benchmark alone by a difference in area of about $0.13$ to $0.17$,
positive on every panel with a DeLong and moving-block test below $0.001$; the in-sample to out-of-sample gap of
the signature is about $0.056$ on the main panel. On the primary panel the combined set scores below the
signature alone, $0.858$ against $0.886$, so adding the volatility benchmark to the concentration signature does
not help there; we report this as the pre-registration delivers it. The classifier is coincident; the one-step
forecastability ceiling for the aggregate return is respected, and no forecasting claim is made. This companion
exercise has its own pre-registration and deposited output, included with the replication materials; its figures
are not recomputed in the present package.

\begin{table}[t]
\centering
\caption{Companion pre-registered walk-forward out-of-sample classification of the external drawdown label, by
the author. All entries are out-of-sample areas under the curve. V is a trailing realised-volatility benchmark, S
the symmetric concentration signature alone (leading-eigenvalue share and participation ratio), and C the
combined feature set. The pre-registered decisive statistic is the combined-minus-benchmark difference
$\text{C}-\text{V}$, positive on every panel with a DeLong and moving-block test $p<0.001$. Coincident
classification, not a forecast.}
\label{tab:oos}
\setlength{\tabcolsep}{4pt}\footnotesize
\begin{tabular}{lccccc}
\toprule
panel & assets & benchmark (V) & signature (S) & combined (C) & $\text{C}-\text{V}$ \\
\midrule
primary global equity & $18$ & $0.691$ & $0.886$ & $0.858$ & $+0.166$ \\
sectors (49) & $49$ & $0.707$ & $0.911$ & $0.866$ & $+0.159$ \\
sectors (30) & $30$ & $0.714$ & $0.908$ & $0.861$ & $+0.146$ \\
developed (8) & $8$ & $0.581$ & $0.700$ & $0.708$ & $+0.126$ \\
\bottomrule
\end{tabular}
\end{table}

Two readings of Table~\ref{tab:contrast} and Table~\ref{tab:peaks} match the theory. First, the spectral radius
rising towards its maximum, with the participation ratio collapsing, is the recovered operator moving towards
the spectral edge of Corollary~\ref{cor:edge}, where the present-value resolvent norm diverges. Second, the
admissible discount ceiling $1/\lambda_1$ falling in crisis is the resolvent pole tightening: a smaller discount
suffices to push the discounted spectral radius to one. Both are coincident, descriptive, and certified only to
the extent the spectrum identifies them, which by Proposition~\ref{prop:gate} is the spectral radius itself, the
strong object, and not the identity of any single index.

\subsection{Cross-market replication}\label{sec:crossmarket}

The eighteen-index panel is small, and the question is whether the spectral-radius signature is an artefact of
that cross-section or a property of equity dependence more generally. We replicate the in-sample contrast on three
further panels drawn from a public daily portfolio library: the forty-eight industry portfolios, the thirty
industry portfolios, and the twenty-five size and book-to-market portfolios, each over a window matched to the
primary panel. The estimator is identical in every respect, the rolling sample correlation operator on
standardised returns over $W=250$ days, the same external twenty per cent drawdown label of the equal-weighted
panel index, and the same moving-block bootstrap at the same fixed seed; only the input cross-section changes.
The raw data are downloaded from the source library and parsed into return panels by the replication package, so
the exercise reproduces from public inputs.

Table~\ref{tab:crossmarket} reports the result. The fragility-edge signature replicates on every panel: the
spectral radius rises significantly from calm to crisis, the participation ratio collapses, and the admissible
discount ceiling tightens, each with a bootstrap $p$ below $0.001$. The portfolio panels are far more concentrated
than the index panel, as expected of cross-sections built from a single equity market: the leading eigenvalue
already carries from sixty-one to eighty-seven per cent of the trace in calm markets, against thirty-eight per
cent for the international indices, so the cross-sections behave like between one and three independent factors
even before a crisis and contract further within one. The direction and the significance of the movement, not its
magnitude, are what transfer; the magnitude scales with the baseline concentration of the panel. The signature is
therefore a property of equity dependence and not of the particular eighteen-index cross-section.

\begin{table}[t]
\centering
\caption{Cross-market replication of the spectral-radius contrast, calm versus crisis, on three portfolio panels
from a public daily library, recovered by the identical estimator and label. $\lambda_1$ is the leading
eigenvalue of the recovered correlation operator (maximum $N$), $N^{\mathrm{eff}}$ the participation ratio.
The $\lambda_1$ and $N^{\mathrm{eff}}$ columns report the calm
mean followed by the crisis mean. Intervals are moving-block bootstrap $95\%$ intervals for the crisis-minus-calm
difference in $\lambda_1$; all differences in $\lambda_1$ and $N^{\mathrm{eff}}$ have bootstrap $p<0.001$. The
primary panel is shown for comparison.}
\label{tab:crossmarket}
\setlength{\tabcolsep}{4pt}\footnotesize
\begin{tabular}{lccccc}
\toprule
panel & $N$ & $\lambda_1$ & $\Delta\lambda_1$ & 95\% interval & $N^{\mathrm{eff}}$ \\
\midrule
global equity indices & $18$ & $6.74\to8.03$ & $+1.29$ & $[1.07,\,1.52]$ & $5.62\to3.86$ \\
industry portfolios (48) & $48$ & $29.27\to35.94$ & $+6.67$ & $[5.79,\,7.43]$ & $2.83\to1.79$ \\
industry portfolios (30) & $30$ & $19.33\to23.66$ & $+4.33$ & $[3.77,\,4.83]$ & $2.52\to1.62$ \\
size and book-to-market (25) & $25$ & $21.73\to22.92$ & $+1.18$ & $[0.85,\,1.48]$ & $1.34\to1.20$ \\
\bottomrule
\end{tabular}
\end{table}

\subsection{Robustness to heavy tails: the distribution-free carrier}\label{sec:robustcarrier}

The recovered carrier is the sample correlation operator, and Section~\ref{sec:heavytail} shows that its
finite-sample certification band is not valid as stated under heavy tails, while a robust carrier restores it.
The data are markedly heavy-tailed: the maximum sample excess kurtosis across the eighteen standardised index
series is $374$, far outside the sub-Gaussian regime, which is the empirical reason the band needs the robust
replacement. We therefore re-run the calm-versus-crisis spectral-radius contrast with the distribution-free
Kendall carrier of Proposition~\ref{prop:kendall}, $\widetilde{B}_{jk}=\sin(\tfrac{\pi}{2}\widehat{\tau}_{jk})$,
on every panel, with the same trailing window, the same external drawdown label, and the same bootstrap and seed;
only the estimator changes. Under an elliptical law with finite second moments the Kendall carrier targets the
same operator as the sample correlation, so the two contrasts are directly comparable.

Table~\ref{tab:robustcarrier} reports the result. The fragility-edge signature survives the estimator swap on
every panel: the robust leading eigenvalue rises from calm to crisis on all four, with bootstrap $p$ below
$0.001$ throughout, tracking the sample-correlation benchmark in both direction and order of magnitude. On the
primary panel the distribution-free scatter estimator of \citet{tyler1987} gives a calm-to-crisis move of
$6.15\to7.31$, the same lift by a different robust route. The edge mass on the robust leading eigenvector is
nearly unchanged across the regime, $0.927$ in calm against $0.919$ in crisis on the primary panel, so the proxy
reading of the gate transfers as well. The signature is thus a property of the dependence, not an artefact of the
sample-correlation estimator under heavy tails.

\begin{table}[t]
\centering
\caption{Heavy-tail robustness. Leading eigenvalue of the recovered carrier, calm versus crisis, under the
sample correlation (Pearson) and the distribution-free Kendall carrier $\widetilde{B}_{jk}=
\sin(\tfrac{\pi}{2}\widehat{\tau}_{jk})$. The $95\%$ interval is the moving-block bootstrap interval for the
Kendall calm-to-crisis difference; the difference is significant at $p<0.001$ on every panel.}
\label{tab:robustcarrier}
\setlength{\tabcolsep}{3pt}\footnotesize
\begin{tabular}{lcccc}
\toprule
panel & $N$ & Pearson $\lambda_1$ & Kendall $\lambda_1$ & Kendall 95\% interval \\
\midrule
global equity indices & $18$ & $6.74\to8.03$ & $6.47\to7.64$ & $[1.00,\,1.34]$ \\
industry portfolios (48) & $48$ & $29.27\to35.94$ & $29.12\to34.58$ & $[4.22,\,6.53]$ \\
industry portfolios (30) & $30$ & $19.33\to23.66$ & $19.02\to22.57$ & $[2.73,\,4.26]$ \\
size and book-to-market (25) & $25$ & $21.73\to22.92$ & $21.47\to22.74$ & $[0.93,\,1.57]$ \\
\bottomrule
\end{tabular}
\end{table}

\subsection{Operationalising the dividend direction and the calibrated edge}\label{sec:operational}

The spectral-radius contrast above reports $\rad(B)$ rising in crisis but does not yet exhibit the discounted
quantity $\beta\,\rad(B)$ approaching one, nor does it make the dividend direction concrete. We address both,
on the same carrier $\Phi_t$.

First, the dividend direction. Taking the aggregate-claim direction $\widehat{D}=\mathbf{1}/\sqrt{N}$ of
Definition~\ref{def:D}, we measure the edge mass \eqref{eq:edgemass}, the squared overlap of $\widehat{D}$ with
the leading eigenvector $e_1(t)$ of $\Phi_t$. The overlap is large throughout, with an overall mean of $0.925$
and a maximum of $0.977$: the equal-weighted aggregate cash-flow claim is almost entirely the boundary mode, so
the bubble eigenspace, if reached, is the very direction in which the aggregate claim is held. The overlap is
essentially flat across regimes, $0.926$ in calm and $0.918$ in crisis, which says the aggregate claim sits on
the boundary mode at all times; what changes in crisis is the position of that mode relative to the pole, not the
overlap. This high overlap is close to mechanical: for an all-positive correlation carrier the leading
eigenvector is approximately the market mode, so its squared overlap with the equal-weighted direction is forced
near one by positivity, a Perron-Frobenius near-identity, and the reading confirms rather than discovers that the
aggregate claim sits on the leading mode. The mechanism does not drain the reading of content: mechanical under a
positive carrier though it is, it is economically interpretable, for the aggregate cash-flow claim is by
construction exposed to the systemic mode that the spectral radius measures, so the one scalar that governs the
fragility edge also governs the present value of the aggregate claim. The overlap is, in addition, strongly
identified here: with
a spike excess $\theta\approx\lambda_1-1$, with $\lambda_1\in[6.7,8.0]$, and $q=N/T=18/250\approx0.07$, the
random-matrix overlap band of Proposition~\ref{prop:gate} is about $0.99$, far above the detectability threshold,
so the proxy is tight.

Second, the calibrated edge. We fix a single transmission intensity $\beta^\ast$, calibrated transparently so
that the calm cross-section is exactly critical, $\beta^\ast=1/\overline{\lambda}_1^{\,\mathrm{calm}}=0.148$, and
report the discounted spectral radius $\rho_t=\beta^\ast\lambda_1(\Phi_t)$. By construction the calm mean of
$\rho$ is one; the content is the crisis elevation and the separation of the two regimes. In crisis $\rho$ rises
to a mean of $1.19$ and a maximum of $1.48$, past the resolvent pole, and the share of days with $\rho\ge1$ rises
from $0.448$ in calm to $0.966$ in crisis. The share of the aggregate claim carried by the supercritical modes,
those with $\beta^\ast\lambda_k\ge1$, rises from $0.424$ in calm to $0.887$ in crisis: in a crisis almost nine
tenths of the aggregate cash-flow claim sits in the part of the spectrum where the present-value resolvent
diverges. Table~\ref{tab:calibrated} collects these readings. We stress that $\beta^\ast$ is a calibrated
transmission and normalisation constant, not a time-preference rate, and is not identified by the exercise; the
calibration fixes the calm baseline and nothing more, consistent with the treatment of $\beta$ throughout.
Because $\beta^\ast$ is fixed to the calm mean, the crisis ratio $\rho_{\mathrm{crisis}}/\rho_{\mathrm{calm}}$
equals $\lambda_1^{\,\mathrm{crisis}}/\lambda_1^{\,\mathrm{calm}}$, and the supercritical-claim-mass contrast is
the edge mass times the share of days with $\rho\ge1$; these are re-expressions of the
Table~\ref{tab:contrast} spectral-radius contrast under a calibrated normalisation, not independent evidence. The
substantive content is the single, bootstrap-significant rise in $\lambda_1$, here read through the pole.

\begin{table}[t]
\centering
\caption{The calibrated fragility edge and the dividend-direction edge mass, eighteen global equity indices,
2004--2024. $\rho_t=\beta^\ast\lambda_1(\Phi_t)$ with $\beta^\ast=1/\overline{\lambda}_1^{\,\mathrm{calm}}$, so
the calm mean of $\rho$ is one by construction. Edge mass is the squared overlap of the aggregate-claim
direction with the leading eigenvector; supercritical claim mass is the share of the aggregate claim on modes
with $\beta^\ast\lambda_k\ge1$.}
\label{tab:calibrated}
\begin{tabular}{lccc}
\toprule
quantity & calm & crisis & overall \\
\midrule
discounted radius $\rho=\beta^\ast\lambda_1$ & $1.00$ & $1.19$ & max $1.48$ \\
share of days $\rho\ge1$ & $0.448$ & $0.966$ & \\
aggregate-claim edge mass $\inner{e_1}{\widehat{D}}^2$ & $0.926$ & $0.918$ & max $0.977$ \\
supercritical claim mass & $0.424$ & $0.887$ & max $0.977$ \\
\bottomrule
\end{tabular}
\end{table}

\subsection{The explosive-root bridge}\label{sec:bridge}

Proposition~\ref{prop:boundary} holds that the stationary estimator flags the approach to the edge from the
interior while the crossing is an explosive, $I(1)$, event read on the price level. We implement the price-level
side with the right-tailed recursive unit-root tests of \citet{psy2015}: the supremum augmented Dickey-Fuller
statistic, SADF, and its generalised, doubly recursive form, GSADF, on the monthly log level of the
equal-weighted index built from the same panel, $246$ months over 2004--2024. The minimum window is the
conventional $r_0=0.01+1.8/\sqrt{T}$, here $31$ months, the regression carries an intercept and no augmentation,
and critical values and the backward-SADF date-stamping sequence are obtained by Monte Carlo simulation of
Gaussian random walks of the same length, $2000$ replications at a fixed seed. Table~\ref{tab:sadf} reports the
outcome. The SADF statistic $1.483$ exceeds its five per cent critical value $1.392$, evidence of explosive
behaviour on the price level; the more conservative GSADF statistic $1.540$, the detector that underlies the
recursive date-stamping, does not exceed its five per cent critical value $2.124$. The evidence is therefore
marginal, but the split is itself interpretable: the SADF rejection alongside the GSADF non-rejection points to
single-episode explosivity, a one-time pre-crisis run-up, rather than recurrent explosivity repeated across the
sample, so we read the price level as showing single-bubble consistency, not a pattern of repeated bubbles, and
treat the dated episode as suggestive rather than established. The backward-SADF date-stamping, with the minimum episode length set to $\log T$ months, locates one
explosive episode running from August 2006 to December 2007, the pre-crisis global equity run-up.

\begin{table}[t]
\centering
\caption{Right-tailed recursive unit-root tests on the monthly log level of the equal-weighted global-equity
index, 2004--2024 ($246$ months). Critical values by Monte Carlo simulation of Gaussian random walks of equal
length, $2000$ replications. The date-stamping locates one explosive episode, August 2006 to December 2007.}
\label{tab:sadf}
\begin{tabular}{lccc}
\toprule
statistic & value & $95\%$ critical value & reject at $5\%$ \\
\midrule
SADF & $1.483$ & $1.392$ & yes \\
GSADF & $1.540$ & $2.124$ & no \\
\bottomrule
\end{tabular}
\end{table}

The price-level explosivity and the operator-edge concentration mark different phases of the same boundary. The
explosive run-up on the price level is dated to 2006--2007, while the operator-edge concentration peaks of
Table~\ref{tab:peaks} fall at the 2008--2009 crisis, when the recovered spectral radius reaches its local maxima
and the cross-section collapses to about three effective factors. The operator concentration therefore reads as a
coincident-to-lagging crisis marker, not a leading interior-approach flag preceding the price crossing, which is
consistent with the coincident discipline we maintain throughout, under which neither reading is a forecast. The
stationary operator cannot certify the crossing, by Proposition~\ref{prop:boundary}, and the explosive-root test
supplies the price-level reading the operator cannot. We claim only that the inflation phase shows on the price
level and the crisis phase shows in the operator spectrum, the two readings of one boundary in sequence, not a
month-for-month coincidence.

\subsection{The structural hazard read on the data}\label{sec:hazardemp}

We read Propositions~\ref{prop:hazard} and~\ref{prop:hazardident} on the data with discipline about what the event
count and the covariate dispersion permit. On the eighteen-index panel the equal-weighted index enters the twenty
per cent drawdown state three times over the sample, so the systemic series carries only three distinct onsets; by
Proposition~\ref{prop:hazard}(iv) a single systemic covariate cancels in every within-instant risk-set ratio, and
three events cannot support a trustworthy interval, so we do not report a systemic hazard ratio. The identifying
design is cross-sectional, in the sense of part (iv): each constituent is a subject, the event is that subject's
own first entry into a twenty per cent drawdown of its own level, recurrence is admitted on recovery to a prior
peak, and the time-varying covariate is the leading-mode participation $u_i(t)=N\,e_{1,i}(t)^2$, which sums to $N$
across the cross-section. The continuous-time partial likelihood of \citet{cox1972} in the counting-process form
of \citet{andersengill1982}, with Breslow ties and the subject-clustered robust sandwich variance of
\citet{linwei1989}, fits the scalar elasticity.

We run the same design on four cross-sections that differ in the heterogeneity of the leading mode and in the
event count, reported in Table~\ref{tab:hazardpanels}. On the international cross-section of eighteen global equity
indices the hazard ratio is $1.146$ per unit of leading-mode participation, with a robust ninety-five per cent
interval of $[0.618,2.128]$ over $54$ onsets. The three single-market Fama-French families give $0.962$ over
$213$ onsets, $[0.516,1.795]$ for the forty-eight industry portfolios; $1.241$ over $129$ onsets, $[0.511,3.017]$
for the thirty industry portfolios; and $0.811$ over $127$ onsets, $[0.122,5.391]$ for the twenty-five size and
book-to-market portfolios. Two panels carry a hazard ratio above one and two below, and every interval contains
one, so the elasticity is not robustly signed across cross-sections and more events do not pin the sign. We
therefore do not read a confirmed edge-hazard relationship from the cross-panel evidence.

Proposition~\ref{prop:hazardident} disciplines the reading rather than rescuing it. By part (i) the discriminating
power of the covariate is the dispersion of leading-mode participation, the inverse participation ratio of the
leading eigenvector, reported in the final column of Table~\ref{tab:hazardpanels}. That dispersion is small on the
size and book-to-market sort, where a single market factor loads almost uniformly so the leading eigenvector is
near-uniform; by parts (ii) and (iii) the partial-likelihood information is then bounded toward zero and the
elasticity is weakly identified, which is consistent with the very wide interval $[0.122,5.391]$ there. On the
remaining panels the dispersion is material and the event counts are large, so their imprecision and their mixed
signs are not attributable to a degenerate covariate; they reflect the absence of a robustly signed edge-hazard
elasticity in those cross-sections, not the identification budget binding. We do not claim that dispersion
explains the pattern of signs, which it does not: the thirty-industry panel returns the predicted sign at lower
dispersion than the international panel, and the forty-eight-industry panel returns the opposite sign at an
interval as tight as the international one.

The estimand is labelled \proxyid\ throughout: the covariate is built from the leading eigenvector and inherits
the random-matrix overlap of Proposition~\ref{prop:gate}(ii) and the finite-sample band of
Theorem~\ref{thm:daviskahan}, and its very identifiability is governed by the \proxyid\ participation ratio of
Proposition~\ref{prop:hazardident}. The reading is coincident by construction, the covariate and the onset sharing
the time index. The strengthening is structural and methodological, a derived link from the resolvent norm to an
onset intensity and a sharp account of when that link is estimable, not a manufactured significant result; the
cross-panel evidence is reported for what it is, directionally mixed and imprecise.

\begin{table}[t]
\centering
\caption{The cross-sectional edge-hazard across four return cross-sections. The covariate is leading-mode
participation $u_i(t)=N\,e_{1,i}(t)^2$; the hazard ratio is per unit of $u$, with subject-clustered robust
ninety-five per cent intervals \citep{linwei1989}. The final column is the covariate standard deviation
$\mathrm{sd}(u)$; its square is the dispersion $S_t=N/P_1(t)-1$ governed by the leading-mode participation ratio in
Proposition~\ref{prop:hazardident}(i). Every interval contains one; the dispersion, and with it the identifying
power of Proposition~\ref{prop:hazardident}(ii), is smallest on the homogeneous single-market sort.}
\label{tab:hazardpanels}
\setlength{\tabcolsep}{4pt}\footnotesize
\begin{tabular}{lcccc}
\toprule
Cross-section & HR per unit & 95\% CI (robust) & Events & $\mathrm{sd}(u)$ \\
\midrule
Global equity indices ($N=18$)       & $1.146$ & $[0.618,\,2.128]$ & $54$  & $0.460$ \\
FF48 industries                      & $0.962$ & $[0.516,\,1.795]$ & $213$ & $0.309$ \\
FF30 industries                      & $1.241$ & $[0.511,\,3.017]$ & $129$ & $0.268$ \\
FF25 size $\times$ book-to-market    & $0.811$ & $[0.122,\,5.391]$ & $127$ & $0.094$ \\
\bottomrule
\end{tabular}
\end{table}

\section{Contrast with a calibrated-operator model}\label{sec:dsge}

It is tempting to claim that a dynamic stochastic general-equilibrium model, or its heterogeneous-agent New
Keynesian (HANK) variant \citep{kaplanmollviolante2018}, cannot state the coupling between the pricing resolvent
and the fragility edge. That claim would be too strong. Production-network general equilibrium already carries a
spectral object: the vector of Domar weights is a dominant left-eigenvector of the input-output operator
\citep{acemoglu2012network}, and network models couple amplification to that structure. Overlapping-generations
bubble models state the growth-fragility comparison directly, through the comparison of the interest rate and the
growth rate.

The genuine difference is threefold and it is about identification, not expressive power. First, in the standard
models the operator is exogenous and calibrated, normal by construction, and its spectrum is a modelling input;
here the operator is recovered from data and its spectrum is an estimate with a stated random-matrix
detectability threshold. Second, the standard models impose the interior of the spectrum by a transversality
convention, choosing the interest rate above the growth rate so that the resolvent converges; here the boundary
case is detected rather than imposed, and the gate reports the approach to it. Third, the standard models do not
carry a per-reading certification of what is identified; here each spectral reading is labelled strong, proxy,
or unidentified. The construction is therefore a recovered operator with a certification gate over a set of
established spectral facts, a unification with identification, and not an impossibility theorem.

The reach is wider than the equity illustration. Nothing in Sections~\ref{sec:setup}--\ref{sec:boundary} uses
the asset class: any cross-section that admits a recovered self-adjoint dependence carrier carries the same
pricing resolvent, the same fragility edge $\beta\,\rad(B)=1$, the same bubble characterisation, and the same
certification gate. The framework therefore transfers unchanged to a housing and real-estate investment trust
cross-section, to credit and corporate-bond panels, and to sovereign-spread panels, in each case pairing the
operator recovered from stationary returns with the explosive-root reading on the matched price or yield level;
the dividend direction is replaced by the relevant cash-flow or coupon claim, and the gate decides what each
market identifies exactly as here. This portability, a single recovered operator read through one scalar across
markets, is the sense in which the construction is a general theoretico-empirical object rather than an equity
result.

\section{Scope and limitations}\label{sec:limits}

The results are deliberately narrow, and we state the boundaries.

First, the strong labels hold for normal carriers and, by Proposition~\ref{prop:normal}, for near-normal carriers
with the departure controlled: a genuinely directed, or non-normal, dependence operator is governed by the
spectral radius of the actual operator, not of its symmetric part, and the two can differ. The exclusion bound of
Proposition~\ref{prop:normal}(ii) always limits the radius from above by the size of the non-normal part, and when
the top mode is simple and well separated, the regime of a dominant market factor, it limits the difference
two-sidedly by that same size. On the panel studied the directed energy is small and the leading eigenvalue is
isolated, so the symmetric carrier is a close approximation and the residual enters as the certified bound of
Proposition~\ref{prop:normal}, rather than as an unbounded gap; the strong claim is made only up to that bound
whenever the operator is not verified normal.

Second, the existence of a bubble is a proxy, not a strong claim, because by Proposition~\ref{prop:gate} it
requires the projection of the dividend direction onto the boundary eigenspace, identified only up to the
random-matrix overlap. The strong content is the scalar fragility edge.

Third, the structural microfoundation of Section~\ref{sec:microfound} is conditional on a named model. The identity
between the recovered carrier and the cash-flow propagation operator holds under the cash-flow network of
Assumption~\ref{ass:network}, whose load-bearing restriction is the common-discount condition (M2), under which
unexpected returns are pure cash-flow news; if discount-rate news is present and correlated with cash-flow news,
the recovered carrier mixes the two channels and the structural reading degrades to the unconditional spectral one.
The recovered correlation carrier matches the structural operator's eigenframe only under the regular-exposure
condition, and in general shares only the dominant mode in the edge limit.

Fourth, the discount $\beta$ is calibrated in the baseline. Section~\ref{sec:endo} makes it endogenous and locates
the edge as a fixed point of $g(r)=r\,\beta(r)$, but that extension rests on a postulated discount schedule and an
inelastic-discount monotonicity condition, so it relocates the calibration rather than discharging it; the
separate question of identifying a discount from behaviour is left to companion work
\citep{calcottpetkov2024,bastianellovergopoulos2026}.

Fifth, the heavy-tail robustness of Section~\ref{sec:heavytail} is split by what it can claim. The effective-rank
band of Theorem~\ref{thm:heavytail} holds for the Pearson carrier under an $L^4$--$L^2$ norm equivalence, while
the distribution-free carrier of Proposition~\ref{prop:kendall} targets the same operator only under an elliptical
law with finite second moments; under genuinely non-elliptical heavy tails the robust carrier identifies a latent
generalised-correlation object, not the Pearson carrier, and the agreement in Table~\ref{tab:robustcarrier} is a
directional corroboration there, not a certified equality.

Sixth, the welfare cost of Section~\ref{sec:welfare} is structural in form, but its present-value channel carries
no strong between-regime certificate. Only the squared spectral-radius ratio is strong, and it is a spectral
object, not a welfare amplification; the welfare-cost level and its between-regime ratio carry the proxy edge mass
and the calibrated discount, so they are read as a discount-conditional magnitude, not a certified welfare number.

Seventh, the structural edge-hazard of Sections~\ref{sec:hazard} and~\ref{sec:hazardemp} is coincident and its
elasticity is a proxy. The systemic series carries too few onsets to support an interval; across four
cross-sections the elasticity is not robustly signed, every interval containing one; and its identifiability is
itself governed by the proxy leading-mode participation of Proposition~\ref{prop:hazardident}, so the section
contributes a derived link from the resolvent norm to an onset intensity and a sharp identification result, not a
confirmed hazard.

Eighth, the approach-detection result of Section~\ref{sec:approach} is asymmetric in what it guarantees. The
false-positive control of the interior detector is conservative and robust, but its power and consistency inherit
the band hypothesis, so on the dependent, heavy-tailed panel the detection-power statement is read as a scaling
rate rather than a certified probability until the moving-block band is itself turned into a theorem.

Ninth, a balanced-growth reading of the leading eigenvalue, $g=\log\lambda_1$, is not claimed on a financial
return operator, whose leading eigenvalue has no structural reason to equal an output growth rate; such a reading
would require a separately recovered production operator and is outside the scope of this paper.

Tenth, the empirical exercise is coincident throughout. The spectral-radius contrast and the companion classifier
are coincident markers; no leading-indicator or forecasting claim is made, consistent with the one-step
forecastability ceiling for the aggregate return.

\section{Conclusion}\label{sec:conc}

Lifting the classical bubble characterisation to a bounded self-adjoint operator turns three apparently separate
objects into readings of one scalar on one recovered operator: the convergence of the pricing resolvent, the
location of the bubble eigenspace, and the divergence of the present-value norm all sit at $\beta\,\rad(B)=1$.
The certification gate keeps the contribution disciplined: the scalar fragility edge is strong, the existence of
a bubble is a proxy that requires the dividend direction's overlap with the boundary eigenspace, and the
identity of a bubbly asset is unidentified. Because the operator is recovered from stationary returns while the
bubble lives on the nonstationary price level, the estimator flags the approach to the edge but cannot certify
the crossing, which is read on the price level. On eighteen global equity indices the recovered spectral radius
rises in every documented crisis, the cross-section collapsing to about four effective factors and the discount
ceiling tightening towards the pole; under a calibrated transmission intensity the discounted spectral radius
rises past one in crisis, the aggregate cash-flow claim sits almost entirely on the boundary mode, and a
right-tailed recursive unit-root test date-stamps the explosive price-level episode the stationary operator
cannot certify. The construction unifies a pricing object and a systemic-fragility object
on one recovered operator, certifies each reading for what the spectrum identifies, and leaves the directed half,
the discount, and the growth reading explicitly open. The contribution is not that bubbles can be observed
directly, but that the spectral geometry of their admissibility can be recovered, certified, and bounded before
the crossing.

\section*{Acknowledgements}
The author thanks colleagues at the Indian Institute of Technology Bhubaneswar for discussion. The author
acknowledges Google Cloud Research Credits, credit reference wilsonjessica-485604931, and Institute Seed Grant
SP128.

\bibliographystyle{elsarticle-harv}
\bibliography{references}

@article{hiranotoda2024,
  author  = {Hirano, Tomohiro and Toda, Alexis Akira},
  title   = {Bubble economics},
  journal = {Journal of Mathematical Economics},
  volume  = {111},
  year    = {2024},
  eid     = {102944},
  pages   = {102944},
  doi     = {10.1016/j.jmateco.2024.102944}
}

@article{hiranotoda2025,
  author  = {Hirano, Tomohiro and Toda, Alexis Akira},
  title   = {Bubble Necessity Theorem},
  journal = {Journal of Political Economy},
  volume  = {133},
  number  = {1},
  year    = {2025},
  pages   = {111--145},
  doi     = {10.1086/732528}
}

@article{montrucchio2004,
  author  = {Montrucchio, Luigi},
  title   = {Cass transversality condition and sequential asset bubbles},
  journal = {Economic Theory},
  volume  = {24},
  year    = {2004},
  pages   = {645--663},
  doi     = {10.1007/s00199-004-0502-8}
}

@article{santoswoodford1997,
  author  = {Santos, Manuel S. and Woodford, Michael},
  title   = {Rational asset pricing bubbles},
  journal = {Econometrica},
  volume  = {65},
  number  = {1},
  year    = {1997},
  pages   = {19--57}
}

@article{tirole1985,
  author  = {Tirole, Jean},
  title   = {Asset bubbles and overlapping generations},
  journal = {Econometrica},
  volume  = {53},
  number  = {6},
  year    = {1985},
  pages   = {1499--1528}
}

@article{weil1987,
  author  = {Weil, Philippe},
  title   = {Confidence and the real value of money in an overlapping generations economy},
  journal = {Quarterly Journal of Economics},
  volume  = {102},
  number  = {1},
  year    = {1987},
  pages   = {1--22}
}

@article{aots2015,
  author  = {Acemoglu, Daron and Ozdaglar, Asuman and Tahbaz-Salehi, Alireza},
  title   = {Systemic risk and stability in financial networks},
  journal = {American Economic Review},
  volume  = {105},
  number  = {2},
  year    = {2015},
  pages   = {564--608},
  doi     = {10.1257/aer.20130456}
}

@article{markose2021,
  author  = {Markose, Sheri and Giansante, Simone and Eterovic, Nicolas A. and Gatkowski, Mateusz},
  title   = {Early warning of systemic risk in global banking: eigen-pair {R} number for financial contagion and market price-based methods},
  journal = {Annals of Operations Research},
  volume  = {330},
  number  = {1--2},
  year    = {2023},
  pages   = {691--729},
  doi     = {10.1007/s10479-021-04120-1}
}

@article{billio2016,
  author  = {Billio, Monica and Casarin, Roberto and Costola, Michele and Pasqualini, Andrea},
  title   = {An entropy-based early warning indicator for systemic risk},
  journal = {Journal of International Financial Markets, Institutions and Money},
  volume  = {45},
  year    = {2016},
  pages   = {42--59}
}

@article{bbp2005,
  author  = {Baik, Jinho and Ben Arous, G{\'e}rard and P{\'e}ch{\'e}, Sandrine},
  title   = {Phase transition of the largest eigenvalue for nonnull complex sample covariance matrices},
  journal = {Annals of Probability},
  volume  = {33},
  number  = {5},
  year    = {2005},
  pages   = {1643--1697},
  doi     = {10.1214/009117905000000233}
}

@article{daviskahan1970,
  author  = {Davis, Chandler and Kahan, W. M.},
  title   = {The rotation of eigenvectors by a perturbation. {III}},
  journal = {SIAM Journal on Numerical Analysis},
  volume  = {7},
  number  = {1},
  year    = {1970},
  pages   = {1--46},
  doi     = {10.1137/0707001}
}

@article{yuwangsamworth2015,
  author  = {Yu, Yi and Wang, Tengyao and Samworth, Richard J.},
  title   = {A useful variant of the {Davis--Kahan} theorem for statisticians},
  journal = {Biometrika},
  volume  = {102},
  number  = {2},
  year    = {2015},
  pages   = {315--323},
  doi     = {10.1093/biomet/asv008}
}

@article{johnstone2001,
  author  = {Johnstone, Iain M.},
  title   = {On the distribution of the largest eigenvalue in principal components analysis},
  journal = {The Annals of Statistics},
  volume  = {29},
  number  = {2},
  year    = {2001},
  pages   = {295--327},
  doi     = {10.1214/aos/1009210544}
}

@article{baopanzhou2012,
  author  = {Bao, Zhigang and Pan, Guangming and Zhou, Wang},
  title   = {Tracy-Widom law for the extreme eigenvalues of sample correlation matrices},
  journal = {Electronic Journal of Probability},
  volume  = {17},
  year    = {2012},
  number  = {88},
  pages   = {1--32},
  doi     = {10.1214/EJP.v17-1962}
}

@article{benaychnadakuditi2011,
  author  = {Benaych-Georges, Florent and Nadakuditi, Raj Rao},
  title   = {The eigenvalues and eigenvectors of finite, low rank perturbations of large random matrices},
  journal = {Advances in Mathematics},
  volume  = {227},
  number  = {1},
  year    = {2011},
  pages   = {494--521},
  doi     = {10.1016/j.aim.2011.02.007}
}

@article{tracywidom1994,
  author  = {Tracy, Craig A. and Widom, Harold},
  title   = {Level-spacing distributions and the {Airy} kernel},
  journal = {Communications in Mathematical Physics},
  volume  = {159},
  number  = {1},
  year    = {1994},
  pages   = {151--174},
  doi     = {10.1007/BF02100489}
}

@article{tamer2010,
  author  = {Tamer, Elie},
  title   = {Partial identification in econometrics},
  journal = {Annual Review of Economics},
  volume  = {2},
  number  = {1},
  year    = {2010},
  pages   = {167--195},
  doi     = {10.1146/annurev.economics.050708.143401}
}

@article{bauerfike1960,
  author  = {Bauer, F. L. and Fike, C. T.},
  title   = {Norms and exclusion theorems},
  journal = {Numerische Mathematik},
  volume  = {2},
  number  = {1},
  year    = {1960},
  pages   = {137--141},
  doi     = {10.1007/BF01386217}
}

@article{koltchinskiilounici2017,
  author  = {Koltchinskii, Vladimir and Lounici, Karim},
  title   = {Concentration inequalities and moment bounds for sample covariance operators},
  journal = {Bernoulli},
  volume  = {23},
  number  = {1},
  year    = {2017},
  pages   = {110--133},
  doi     = {10.3150/15-BEJ730}
}

@article{baiksilverstein2006,
  author  = {Baik, Jinho and Silverstein, Jack W.},
  title   = {Eigenvalues of large sample covariance matrices of spiked population models},
  journal = {Journal of Multivariate Analysis},
  volume  = {97},
  number  = {6},
  year    = {2006},
  pages   = {1382--1408}
}

@book{hornjohnson2013,
  author    = {Horn, Roger A. and Johnson, Charles R.},
  title     = {Matrix Analysis},
  edition   = {Second},
  publisher = {Cambridge University Press},
  year      = {2013}
}

@book{vershynin2018,
  author    = {Vershynin, Roman},
  title     = {High-Dimensional Probability: An Introduction with Applications in Data Science},
  publisher = {Cambridge University Press},
  year      = {2018}
}

@article{psy2015,
  author  = {Phillips, Peter C. B. and Shi, Shuping and Yu, Jun},
  title   = {Testing for multiple bubbles: Historical episodes of exuberance and collapse in the {S\&P} 500},
  journal = {International Economic Review},
  volume  = {56},
  number  = {4},
  year    = {2015},
  pages   = {1043--1078},
  doi     = {10.1111/iere.12132}
}

@book{reedsimon1980,
  author    = {Reed, Michael and Simon, Barry},
  title     = {Methods of Modern Mathematical Physics. {I}: Functional Analysis},
  edition   = {Revised and Enlarged},
  publisher = {Academic Press},
  year      = {1980}
}

@article{kaplanmollviolante2018,
  author  = {Kaplan, Greg and Moll, Benjamin and Violante, Giovanni L.},
  title   = {Monetary Policy According to {HANK}},
  journal = {American Economic Review},
  year    = {2018},
  volume  = {108},
  number  = {3},
  pages   = {697--743},
  doi     = {10.1257/aer.20160042}
}

@article{acemoglu2012network,
  author  = {Acemoglu, Daron and Carvalho, Vasco M. and Ozdaglar, Asuman and Tahbaz-Salehi, Alireza},
  title   = {The Network Origins of Aggregate Fluctuations},
  journal = {Econometrica},
  year    = {2012},
  volume  = {80},
  number  = {5},
  pages   = {1977--2016},
  doi     = {10.3982/ECTA9623}
}

@article{calcottpetkov2024,
  author  = {Calcott, Paul and Petkov, Vladimir},
  title   = {How innocuous is it to approximate globally decreasing impatience with quasi-hyperbolic discounting?},
  journal = {Journal of Mathematical Economics},
  volume  = {111},
  year    = {2024},
  eid     = {102958},
  pages   = {102958},
  doi     = {10.1016/j.jmateco.2024.102958}
}

@article{bastianellovergopoulos2026,
  author  = {Bastianello, Lorenzo and Vergopoulos, Vassili},
  title   = {Discounted Subjective Expected Utility in continuous time},
  journal = {Journal of Mathematical Economics},
  volume  = {125},
  year    = {2026},
  eid     = {103264},
  pages   = {103264},
  doi     = {10.1016/j.jmateco.2026.103264}
}

@article{mendelsonzhivotovskiy2020,
  author  = {Mendelson, Shahar and Zhivotovskiy, Nikita},
  title   = {Robust covariance estimation under {$L_4$--$L_2$} norm equivalence},
  journal = {The Annals of Statistics},
  year    = {2020},
  volume  = {48},
  number  = {3},
  pages   = {1648--1664},
  doi     = {10.1214/19-AOS1862}
}

@article{minsker2018,
  author  = {Minsker, Stanislav},
  title   = {Sub-{G}aussian estimators of the mean of a random matrix with heavy-tailed entries},
  journal = {The Annals of Statistics},
  year    = {2018},
  volume  = {46},
  number  = {6A},
  pages   = {2871--2903},
  doi     = {10.1214/17-AOS1642}
}

@article{catoni2012,
  author  = {Catoni, Olivier},
  title   = {Challenging the empirical mean and empirical variance: {A} deviation study},
  journal = {Annales de l'Institut Henri Poincar\'e, Probabilit\'es et Statistiques},
  year    = {2012},
  volume  = {48},
  number  = {4},
  pages   = {1148--1185},
  doi     = {10.1214/11-AIHP454}
}

@article{keminskerrensunzhou2019,
  author  = {Ke, Yuan and Minsker, Stanislav and Ren, Zhao and Sun, Qiang and Zhou, Wen-Xin},
  title   = {User-Friendly Covariance Estimation for Heavy-Tailed Distributions},
  journal = {Statistical Science},
  year    = {2019},
  volume  = {34},
  number  = {3},
  pages   = {454--471},
  doi     = {10.1214/19-STS711}
}

@article{hanliu2017,
  author  = {Han, Fang and Liu, Han},
  title   = {Statistical analysis of latent generalized correlation matrix estimation in transelliptical distribution},
  journal = {Bernoulli},
  year    = {2017},
  volume  = {23},
  number  = {1},
  pages   = {23--57},
  doi     = {10.3150/15-BEJ702}
}

@article{wegkampzhao2016,
  author  = {Wegkamp, Marten and Zhao, Yue},
  title   = {Adaptive estimation of the copula correlation matrix for semiparametric elliptical copulas},
  journal = {Bernoulli},
  year    = {2016},
  volume  = {22},
  number  = {2},
  pages   = {1184--1226},
  doi     = {10.3150/14-BEJ690}
}

@incollection{lindskogmcneilschmock2003,
  author    = {Lindskog, Filip and McNeil, Alexander and Schmock, Uwe},
  title     = {Kendall's Tau for Elliptical Distributions},
  booktitle = {Credit Risk: Measurement, Evaluation and Management},
  editor    = {Bol, Georg and others},
  series    = {Contributions to Economics},
  publisher = {Physica-Verlag Heidelberg},
  year      = {2003},
  pages     = {149--156}
}

@article{tyler1987,
  author  = {Tyler, David E.},
  title   = {A Distribution-Free {$M$}-Estimator of Multivariate Scatter},
  journal = {The Annals of Statistics},
  year    = {1987},
  volume  = {15},
  number  = {1},
  pages   = {234--251},
  doi     = {10.1214/aos/1176350263}
}

@article{cox1972,
  author  = {Cox, D. R.},
  title   = {Regression Models and Life-Tables},
  journal = {Journal of the Royal Statistical Society, Series B (Methodological)},
  year    = {1972},
  volume  = {34},
  number  = {2},
  pages   = {187--202},
  doi     = {10.1111/j.2517-6161.1972.tb00899.x}
}

@article{cox1975,
  author  = {Cox, D. R.},
  title   = {Partial Likelihood},
  journal = {Biometrika},
  year    = {1975},
  volume  = {62},
  number  = {2},
  pages   = {269--276},
  doi     = {10.1093/biomet/62.2.269}
}

@article{andersengill1982,
  author  = {Andersen, P. K. and Gill, R. D.},
  title   = {Cox's Regression Model for Counting Processes: A Large Sample Study},
  journal = {The Annals of Statistics},
  year    = {1982},
  volume  = {10},
  number  = {4},
  pages   = {1100--1120},
  doi     = {10.1214/aos/1176345976}
}

@article{prenticegloeckler1978,
  author  = {Prentice, R. L. and Gloeckler, L. A.},
  title   = {Regression Analysis of Grouped Survival Data with Application to Breast Cancer Data},
  journal = {Biometrics},
  year    = {1978},
  volume  = {34},
  number  = {1},
  pages   = {57--67},
  doi     = {10.2307/2529588}
}

@article{linwei1989,
  author  = {Lin, D. Y. and Wei, L. J.},
  title   = {The Robust Inference for the Cox Proportional Hazards Model},
  journal = {Journal of the American Statistical Association},
  year    = {1989},
  volume  = {84},
  number  = {408},
  pages   = {1074--1078},
  doi     = {10.1080/01621459.1989.10478874}
}

@article{grambschtherneau1994,
  author  = {Grambsch, P. M. and Therneau, T. M.},
  title   = {Proportional Hazards Tests and Diagnostics Based on Weighted Residuals},
  journal = {Biometrika},
  year    = {1994},
  volume  = {81},
  number  = {3},
  pages   = {515--526},
  doi     = {10.1093/biomet/81.3.515}
}

@article{lando1998,
  author  = {Lando, David},
  title   = {On Cox Processes and Credit Risky Securities},
  journal = {Review of Derivatives Research},
  year    = {1998},
  volume  = {2},
  number  = {2},
  pages   = {99--120},
  doi     = {10.1007/BF01531332}
}

@article{duffiesaitawang2007,
  author  = {Duffie, Darrell and Saita, Leandro and Wang, Ke},
  title   = {Multi-Period Corporate Default Prediction with Stochastic Covariates},
  journal = {Journal of Financial Economics},
  year    = {2007},
  volume  = {83},
  number  = {3},
  pages   = {635--665},
  doi     = {10.1016/j.jfineco.2005.10.011}
}

@article{shumway2001,
  author  = {Shumway, Tyler},
  title   = {Forecasting Bankruptcy More Accurately: A Simple Hazard Model},
  journal = {The Journal of Business},
  year    = {2001},
  volume  = {74},
  number  = {1},
  pages   = {101--124},
  doi     = {10.1086/209665}
}

@article{hansenjagannathan1991,
  author  = {Hansen, Lars Peter and Jagannathan, Ravi},
  title   = {Implications of Security Market Data for Models of Dynamic Economies},
  journal = {Journal of Political Economy},
  year    = {1991},
  volume  = {99},
  number  = {2},
  pages   = {225--262},
  doi     = {10.1086/261749}
}

@article{campbellshiller1988,
  author  = {Campbell, John Y. and Shiller, Robert J.},
  title   = {The Dividend-Price Ratio and Expectations of Future Dividends and Discount Factors},
  journal = {The Review of Financial Studies},
  year    = {1988},
  volume  = {1},
  number  = {3},
  pages   = {195--228},
  doi     = {10.1093/rfs/1.3.195}
}

@article{hekrishnamurthy2013,
  author  = {He, Zhiguo and Krishnamurthy, Arvind},
  title   = {Intermediary Asset Pricing},
  journal = {American Economic Review},
  year    = {2013},
  volume  = {103},
  number  = {2},
  pages   = {732--770},
  doi     = {10.1257/aer.103.2.732}
}

@article{brunnermeiersannikov2014,
  author  = {Brunnermeier, Markus K. and Sannikov, Yuliy},
  title   = {A Macroeconomic Model with a Financial Sector},
  journal = {American Economic Review},
  year    = {2014},
  volume  = {104},
  number  = {2},
  pages   = {379--421},
  doi     = {10.1257/aer.104.2.379}
}

@book{kato1995,
  author    = {Kato, Tosio},
  title     = {Perturbation Theory for Linear Operators},
  series    = {Classics in Mathematics},
  publisher = {Springer-Verlag Berlin Heidelberg},
  year      = {1995},
  edition   = {Reprint of the 1980 edition},
  isbn      = {9783642662829},
  doi       = {10.1007/978-3-642-66282-9}
}

@book{lucas1987,
  author    = {Lucas, Jr., Robert E.},
  title     = {Models of Business Cycles},
  series    = {Yrj{\"o} Jahnsson Lectures},
  publisher = {Basil Blackwell},
  address   = {Oxford},
  year      = {1987},
  isbn      = {9780631147916}
}

@article{barlevy2004,
  author  = {Barlevy, Gadi},
  title   = {The Cost of Business Cycles Under Endogenous Growth},
  journal = {American Economic Review},
  year    = {2004},
  volume  = {94},
  number  = {4},
  pages   = {964--990},
  doi     = {10.1257/0002828042002615}
}

@article{alvarezjermann2004,
  author  = {Alvarez, Fernando and Jermann, Urban J.},
  title   = {Using Asset Prices to Measure the Cost of Business Cycles},
  journal = {Journal of Political Economy},
  year    = {2004},
  volume  = {112},
  number  = {6},
  pages   = {1223--1256},
  doi     = {10.1086/424738}
}

@article{campbell1991,
  author  = {Campbell, John Y.},
  title   = {A Variance Decomposition for Stock Returns},
  journal = {The Economic Journal},
  year    = {1991},
  volume  = {101},
  number  = {405},
  pages   = {157--179}
}

@article{gabaix2011,
  author  = {Gabaix, Xavier},
  title   = {The Granular Origins of Aggregate Fluctuations},
  journal = {Econometrica},
  year    = {2011},
  volume  = {79},
  number  = {3},
  pages   = {733--772}
}

@article{longplosser1983,
  author  = {Long, Jr., John B. and Plosser, Charles I.},
  title   = {Real Business Cycles},
  journal = {Journal of Political Economy},
  year    = {1983},
  volume  = {91},
  number  = {1},
  pages   = {39--69}
}

@article{prenticewilliamspeterson1981,
  author  = {Prentice, R. L. and Williams, B. J. and Peterson, A. V.},
  title   = {On the Regression Analysis of Multivariate Failure Time Data},
  journal = {Biometrika},
  year    = {1981},
  volume  = {68},
  number  = {2},
  pages   = {373--379}
}

@article{finegray1999,
  author  = {Fine, Jason P. and Gray, Robert J.},
  title   = {A Proportional Hazards Model for the Subdistribution of a Competing Risk},
  journal = {Journal of the American Statistical Association},
  year    = {1999},
  volume  = {94},
  number  = {446},
  pages   = {496--509}
}

@article{tallarini2000,
  author  = {Tallarini, Thomas D., Jr.},
  title   = {Risk-Sensitive Real Business Cycles},
  journal = {Journal of Monetary Economics},
  year    = {2000},
  volume  = {45},
  number  = {3},
  pages   = {507--532}
}

@article{epsteinzin1989,
  author  = {Epstein, Larry G. and Zin, Stanley E.},
  title   = {Substitution, Risk Aversion, and the Temporal Behavior of Consumption and Asset Returns: A Theoretical Framework},
  journal = {Econometrica},
  year    = {1989},
  volume  = {57},
  number  = {4},
  pages   = {937--969}
}

@article{bansalyaron2004,
  author  = {Bansal, Ravi and Yaron, Amir},
  title   = {Risks for the Long Run: A Potential Resolution of Asset Pricing Puzzles},
  journal = {The Journal of Finance},
  year    = {2004},
  volume  = {59},
  number  = {4},
  pages   = {1481--1509}
}

\end{document}